\documentclass[preprint,1p,times]{elsarticle}
\usepackage{geometry}
\usepackage{graphicx} 
\usepackage{chemformula} 
\usepackage[T1]{fontenc} 
\usepackage{natbib}
\usepackage{amsmath,amsthm,amssymb}
\usepackage{mathabx}
\usepackage{float}
\usepackage{graphicx}
\usepackage{caption}
\usepackage{subfigure}
\usepackage{tabularx}
\usepackage{threeparttable}
\usepackage{dcolumn}
\usepackage{bm}
\usepackage{color,epsfig,multirow}
\usepackage{array}
\usepackage{hyperref}
\usepackage{algorithmic}
\usepackage{algorithm}
\usepackage{booktabs}
\usepackage{threeparttable}
\usepackage{makecell}
\usepackage{lyu}

\mathchardef\mhyphen="2D

\def\mb{\mathbf}

\def\bm{\boldsymbol}
\newcommand{\Rmnum}[1]{\uppercase\expandafter{\romannumeral #1\relax}}
\newcommand{\rmnum}[1]{\lowercase\expandafter{\romannumeral #1\relax}}

\newtheorem{theorem}{Theorem}

\theoremstyle{definition}

\theoremstyle{remark}
\newtheorem{remark}[theorem]{Remark}
\newtheorem{proposition}{Proposition}

\allowdisplaybreaks[4]

\begin{document}

\begin{frontmatter}

\title{
A unified framework for data-driven construction of stochastic reduced models with state-dependent memory \\
} 

\author[1]{Zhiyuan She\fnref{fn1}}
\author[1,2]{Liyao Lyu\fnref{fn2}}
\ead{lyuliyao@math.ucla.edu}
\fntext[fn1, fn2]{These authors contributed equally to this work. }
\author[3,4]{Smith Ronain Bryan}
\author[1,5]{Huan Lei}
\ead{leihuan@msu.edu}

\address[1]{{Department of Computational Mathematics, Science \& Engineering, Michigan State University}, 428 S Shaw Ln, East Lansing, MI, USA}
\address[2]{{Department of Mathematics, University of California, Los Angeles}, 520 Portola Plaza, Los Angeles, CA, USA}
\address[3]{{Department of Biomedical Engineering, Michigan State University}, 775 Woodlot Drive, East Lansing, MI, USA}
\address[4]{{Institute for Quantitative Health Science \& Engineering, Michigan State University}, 775 Woodlot Drive, East Lansing, MI, USA}
\address[5]{{Department of Statistics \& Probability, Michigan State University}, 619 Red Cedar Road, East Lansing, MI, USA}

\date{\today}


\begin{abstract}
We present a unified framework for the data-driven construction of stochastic reduced models with state-dependent memory for high-dimensional Hamiltonian systems. The method addresses two key challenges: (\rmnum{1}) accurately modeling heterogeneous non-Markovian effects where the memory function depends on the coarse-grained (CG) variables beyond the standard homogeneous kernel, and (\rmnum{2}) efficiently exploring the phase space to sample both equilibrium and dynamical observables for reduced model construction. Specifically, we employ a consensus-based sampling method to establish a shared sampling strategy that enables simultaneous construction of the free energy function and collection of conditional two-point correlation functions used to learn the state-dependent memory. The reduced dynamics is formulated as an extended Markovian system, where a set of auxiliary variables, interpreted as non-Markovian features, is jointly learned to systematically approximate the memory function using only two-point statistics. The constructed model yields a generalized Langevin-type formulation with an invariant distribution consistent with the full dynamics. We demonstrate the effectiveness of the proposed framework on a two-dimensional CG model of an alanine dipeptide molecule. Numerical results on the transition dynamics between metastable states show that accurately capturing state-dependent memory is essential for predicting non-equilibrium kinetic properties, whereas the standard generalized Langevin model with a homogeneous kernel exhibits significant discrepancies.
\end{abstract}


\end{frontmatter}


\section{Introduction}

The construction of reduced models for high-dimensional dynamical systems poses a fundamental problem in computational mathematics and multiscale modeling. In many real applications related to fluid physics, materials science, and molecular modeling, the direct simulation of the full-dimensional micro-models often becomes prohibitively expensive. Existing approaches often resort to modeling the evolution of a set of coarse-grained (CG) variables to capture the collective dynamics on the resolved scale of interest. Mathematically, this can be viewed as projecting the full dynamics onto a low-dimensional space, yielding an effective reduced model that encodes the effects of the unresolved dynamics. While the heterogeneous multiscale method \cite{HMM_2003, HMM_review_2007} and equation-free \cite{Equation_Free_2003, Equation_Free_2009} provide rigorous frameworks for bridging micro- and macro-scales, these approaches generally rely on a clear scale separation. For Hamiltonian systems without scale separation, the reduced dynamics generally exhibits memory effects and stochastic behavior due to the unresolved degrees of freedom. For such cases, the generalized Langevin equation (GLE) provides a mathematically principled structure for the reduced model. This model can be derived from the Mori–Zwanzig (MZ) projection formalism \cite{Mori1965, Zwanzig61} and describes the effective dynamics of the CG variables in terms of a conservative force, a memory kernel, and a stochastic noise term that are coupled through a fluctuation–dissipation relation to ensure the thermodynamic consistency.

Despite its formal appeal, the practical construction of GLE-type reduced models presents two major challenges. The first lies in the complexity of the reduced model formulation. In particular, the rigorous MZ-based studies \cite{Darve_PNAS_2009,hijon2010mori, Lei_Cas_2010,Vroylandt_Monmarche_JCP_2022,ayaz2022generalized,lyu2023construction} show that the memory term may depend on the CG variables in a complex way where the analytical form is generally unknown.  To construct the closed-form reduced models for stochastic simulations, most existing approaches further simplify the memory term into a homogeneous kernel independent of the CG variables, namely, the standard GLE. Several approaches based on the hierarchical construction \cite{LiXian2014, li2015incorporation, Ma2016, zhu2018faber, ma2019coarse, Hudson2018, zhu2020generalized} and data-driven parameterization \cite{lange2006collective, ceriotti2009langevin, baczewski2013numerical, Lei_Li_PNAS_2016, russo2019deep, Jung_Hanke_JCTC_2017, lee2019multi, Grogan_Lei_JCP_2020, Klippenstein_Vegt_JCP_2021,  vroylandt2022likelihood, she2023data, Zhu_Tang_JCP_2023, xie_Car_E_PNAS_2024, Xie_E_JCTC_2024} have been proposed to construct the homogeneous memory kernel. While such models capture non-Markovian effects beyond mean-field approximations, they are generally insufficient to capture the heterogeneous nature of the energy dissipation processes over the CG space. To address this issue, our recent work \cite{ge2024data} proposed a data-driven approach for constructing the GLE with a state-dependent memory term and showed that this broadly over-simplified effect can play a crucial role in accurately predicting the non-equilibrium transition dynamics, where the standard GLE fails to reproduce.
On the other hand, that study focused on a one-dimensional CG variable, and the approach relies on evaluating three-point correlation functions conditioned on various initial states, which can become computationally expensive for systems with multi-dimensional CG variables.

The second major challenge lies in the efficient exploration and sampling of the phase space with complex energy landscapes. In particular, the energy landscape may consist of multiple metastable states separated by large energy barriers; direct sampling could be trapped in a local minimum. This difficulty has been well recognized for the construction of the free energy function. Several enhanced sampling approaches, such as the umbrella sampling \cite{Torrie_Valleau_Umbrella_JCP_1977, Kumar_Kollman_JCC_1992}, metadynamics \cite{Laio_Parrinello_PNAS_2002}, adaptive biasing force \cite{Darve_JCP_2001}, and temperature accelerated molecular dynamics \cite{abrams2008efficient, TAMD_CPL_2006} have been developed. However, the construction of the state-dependent memory term further requires sampling dynamical observables, such as time correlation functions over the CG space. This makes many biased-sampling approaches less suitable. In particular, existing strategies based on the direct estimation of the probability density function (PDF) of the CG variables or thermodynamic perturbations cannot be readily applied. This calls for a sampling strategy that is capable of concurrently collecting both equilibrium samples for free energy construction and unbiased dynamical data for memory estimation in a unified and consistent manner.

In this work, we aim to address the above two challenges through a unified framework for data-driven construction of stochastic reduced models with complex free energy landscapes and state-dependent memory. Specifically, we employ a consensus-based enhanced sampling method \cite{lyu2023consensus} to facilitate exploration of the phase space. 
Unlike most existing approaches for the free energy construction, the present method is gradient-free where the sampling dynamics is driven by the local residual and therefore independent of the underlying potential of the full model. This feature enables the efficient sampling across stiff energy landscapes and, importantly, supports the concurrent acquisition of equilibrium data (for free energy estimation) and dynamical observables (for memory construction) from a shared set of simulation trajectories. To construct the reduced model, we adopt a data-driven approach through the joint learning of both a set of non-Markovian features and the extended Markovian dynamics of both the CG variables and these features. In contrast to the formulation in Ref. \cite{she2023data}, the proposed model enables systematic approximation of state-dependent memory, while introducing a coherent multiplicative noise term that preserves the consistent invariant distribution. Furthermore, the parameterization relies only on two-point correlation functions, in contrast to the three-point statistics required in Ref. \cite{ge2024data}.
We demonstrate the effectiveness of the proposed approach by constructing a two-dimensional reduced model of an alanine dipeptide molecule in aqueous solution. Numerical results show that the constructed model accurately reproduces correlation functions conditioned on various initial states. More importantly, comparison of the predicted transition times with full MD simulations shows that accurately modeling the broadly over-simplified state-dependent memory is essential for modeling non-equilibrium collective behaviors, whereas the standard GLE with a homogeneous kernel shows limitations. 

\section{Methods}

\subsection{Preliminaries}

The full model under consideration is a Hamiltonian system with $2N\mhyphen$dimensional phase space vector $\bZ =  [\bQ;\bP]$, where $\bQ, \mb P\in \mathbb{R}^{N}$ represent the position and momentum vector, respectively. With Hamiltonian $H(\mb Z) = \frac{1}{2} \mb P^T \mb M^{-1} \mb P + V(\mb Q)$, the governing equation follows 
\begin{equation}
\dot{\mb Z} = \mathbf{S} \nabla H(\mb Z), \quad \mb Z(0) = \mb Z_0,  
\end{equation} 
where $\mb M = {\rm diag}(\mb M_1, \cdots, \mb M_N)$ is the constant mass matrix, $V: \mathbb{R}^N \to \mathbb{R}$ is the potential function and $\mathbf{S}$ is the symplectic matrix.  Our goal is to construct a reduced stochastic model that captures the effective dynamics of a set of CG variables $\mb z = [\mb q; \mb p]$ by defining a mapping $\phi :\mathbb{R}^{2N} \to \mathbb{R}^{2m}$, where $\mb q = \phi^{q}(\mb Q)$ and $\mb p = \phi^{q}(\mb Q, \mb P)$ represent the reaction coordinates and the corresponding momenta, respectively. By the Koopman operator, the time evolution of $\mb z$ can be written as $\dot {\mb z} = \mathcal{L} \mb z$, where $\mathcal{L} \phi(\mb Z) = -    (\left(\nabla H (\mb Z_0) \right)^\top S\nabla_{\mb Z_0}) \phi(\mb Z)$ represents the Liouville operator that depends on the full-dimensional phase vector $\mb Z$. 

To construct the closed form of the reduced model, we follow Zwanzig's projection formalism by defining the projection operator that acts on a phase-space function as a condition expectation with respect to the CG variables $\mb z(0)= \bm z$, i.e.,
\[
\mathcal P_{\bm z} f := \mathbb{E} [ f(\mb Z) \vert \phi(\mb Z) = \bm z] = \frac{\int \delta(\phi(\bZ)-\bm z)f(\bZ)\rho_0(\bZ) \intd \bZ}{\int \delta(\phi(\bZ)-\bz)\rho_0(\bZ) \intd \bZ},
\]
where $\rho_0(\mb Z) \propto \exp(-\beta H(\mb Z))$ represents the equilibrium distribution.  By the Duhamel-Dyson formula, the reduced dynamics follows:
\[
\dot{\mb z}(t) = {\rm e}^{\mathcal{L}t}\mathcal{P}_{\bm z}\mathcal{L}\mb z(0) + \int_0^t \diff s {\rm e}^{\mathcal{L}(t-s)}\mathcal{P}_{\bm z} \mathcal{L} {\rm e}^{\mathcal{Q}_{\bm z}\mathcal{L}s} \mathcal{Q}_{\bm z}\mathcal{L}\mb z(0)   + {\rm e}^{\mathcal{Q}_{\bm z}\mathcal{L}t} \mathcal{Q}_{\bm z}\mathcal{L}\mb z(0),
\label{eq:Z_MZ_full}
\]
where $\mathcal{Q}_{\bm z} = \mb I - \mathcal{P}_{\bm z}$ and the right-hand-side represents the mean-field, the memory and the fluctuation term that arises from the initial condition, respectively. By assuming that the memory term only depends on the CG coordinates $\mb q$ and interpreting the fluctuation terms as a stochastic process, the reduced dynamics can be simplified as
\begin{equation}\label{equ:GLE}
    \begin{aligned}
    \dot \bq &=  \nabla_{\mb p} F(\mb p, \mb q),\\
    \dot \bp &= -\nabla_{\mb q} F(\mb p, \mb q) + \int_0^t \mb K(\mb q(s), t-s)  \dot{\mb q}(s) \diff s + \mathcal{R}_t,
\end{aligned}
\end{equation}
where $F(\mb q, \mb p)$ is the total free energy, $\mb K(\mb q, t)$ is the memory kernel, and $\mathcal{R}_t$ is the multiplicative noise term related to the memory kernel through the fluctuation-dissipation relation.
Specifically, the total free energy $F(\mb q, \mb p)$ takes the form
\begin{equation*}
F(\mb q, \mb p) = U(\mb q) + \frac{1}{2} \mb p^\top \mb m(\mb q)^{-1} \mb p  + \frac{1}{2} \beta^{-1} \ln {\det \left(\mb m(\mb q)\right)},
\end{equation*}
where $U(\mb q)$ is the conservation free energy with respect to $\mb q$ and is determined by the marginal density function $\rho_q(\mb q)$, i.e., 
\begin{equation}
U(\mb q) = -\beta^{-1} \log \rho_q(\mb q), \quad  \rho_q(\mb q) = \int \delta(\phi_q(\mb Q) - \mb q) \rho_0(\mb Q) \diff \mb Q.  
\label{eq:free_energy}
\end{equation}
$\mb m(\mb q)$ is the position-dependent mass matrix given by
\begin{equation}
\mb m(\mb q)^{-1} = \beta \int \dot{\phi}_q(\mb Q)  \dot{\phi}_q(\mb Q)^\top \delta(\phi_q(\mb Q) - \mb q) \rho_0(\mb Q) \diff \mb Q. 
\label{eq:mass}
\end{equation}
The memory kernel $\mb K(\mb q, t)$ is defined by the Zwanzig projection operator $\mathcal{P}_{\bm z}$, i.e.,
\begin{equation}
\mb K(\mb q, t) = \mathcal{P}_{\bm z}\left[({\rm e}^{\mathcal{Q}_{\bm z} \mathcal{L} t}\mathcal{Q}_{\bm z} \mathcal{L}\mb p) (\mathcal{Q}_{\bm z} \mathcal{L}\mb p)^\top\right].    
\label{eq:memory}
\end{equation}
In particular, by simplifying $\mb K(\mb q, t)$ into a homogeneous kernel $\mb K(\mb q, t) \approx \bm\theta (t)$, Eq. \eqref{equ:GLE} reduces to the standard GLE. In this work, we do not take this assumption and aim to accurately model the state-dependent form $\mb K(\mb q, t)$ to capture the heterogeneous energy dissipation in the CG space.  Accordingly, the construction of the reduced model \eqref{equ:GLE} lies in two essential challenges: (\Rmnum{1}) dynamical exploration of the complex energy landscape for efficient sampling of CG observables following $\rho_q(\mb q)$ in Eqs \eqref{eq:free_energy}\eqref{eq:mass}\eqref{eq:memory};  (\Rmnum{2}) accurate approximation of the memory kernel \eqref{eq:memory} preserving the state-dependent non-Markovian effects with a coherent noise term.

\subsection{Consensus-based enhanced sampling}
\label{sec:CAS}

To efficiently explore the CG space and sample both equilibrium and dynamical observables, we adopt the consensus-based adaptive sampling (CAS) method developed in Ref.~\cite{lyu2023consensus}. This approach formulates sampling as a residual-driven minimax optimization problem and evolves an interacting particle system according to a McKean-type stochastic differential equation (SDE). Unlike conventional enhanced sampling methods \cite{Laio_Parrinello_PNAS_2002,TAMD_CPL_2006}, the dynamical exploration of the CAS method is gradient-free 
and independent of the underlying full potential function. Furthermore,  unlike the adaptive sampling strategies \cite{tang2022adaptive, tang2023adversarial, Gao_Yan_SIAM_2023, Gao_Wang_JCP_2023} for solving high-dimensional partial differential equations, the method does not rely on the free query of arbitrary sampling points and is well-suited for systems with constrained phase space where thermodynamically accessible regions are unknown \emph{a priori}.

Without loss of generality, let $X(\mb q)$ denote an quantity of interest and $X_\mathcal{N}(\mb q)$ its surrogate. The CAS method iteratively solves the following minimax problem
\begin{equation}
\label{equ:min-max}
\min_{X_\mathcal{N}}\max_{\rho} (\mathcal{L}_\mathcal{N},\rho),
\end{equation}
where the residual is defined as $\mathcal{L}_\mathcal{N}(\mb q) = \Vert X(\mb q) - X_\mathcal{N}(\mb q)\Vert^2$, and the maximization objective represents the weighted value under the sampling distribution $\rho(\mb q)$, i.e., 
\begin{equation}
    (\mathcal{L}_\mathcal{N},\rho) =  \int (\mathcal{L}_\mathcal{N}(\mathbf{q}) - \kappa_h^{-1} \ln \rho (\mathbf{q})) \rho(\mathbf{q}) \mathrm{d} \mathbf{q}. 
\label{eq:entropy_regular}    
\end{equation}
Here, the second term serves as an entropy regularization. The rationale of this term can be understood by examining the max-problem of Eq. \eqref{equ:min-max}, which is convex with the analytical solution 
\begin{equation}
\rho^{\ast}(\mb q) := \arg \max_{\rho} (\mathcal{L}_\mathcal{N},\rho) \propto \exp(-\kappa_h\mathcal{L}_\mathcal{N}^{-}(\mb q)),
\label{eq:max_problem}
\end{equation}
where $\mathcal{L}_\mathcal{N}^{-}(\mb q) = - \mathcal{L}_\mathcal{N}(\mb q)$ and $\kappa_h$ can be interpreted as the inverse of the temperature. For the low temperature limit ($\kappa_h \to \infty$), $\rho^{\ast}$ 
concentrates as a Dirac measure at $\mb q^{\ast} = \arg\max \mathcal{L}_\mathcal{N}(\mb q)$, suppressing the dynamical exploration.  
Conversely, the high-temperature limit ($\kappa_h \to 0$) yields a more diffuse distribution and enables the efficient exploration of the uncharted sample space.

In this work, we choose $X(\mb q) = -\nabla U(\mb q) $ i.e., the mean force for the free energy, which naturally targets kinetically important regions. Moreover, this choice also enables concurrent collection of data for dynamical observables such that time correlation functions conditioned with various initial states can be conveniently collected; see Remark \ref{rem:adaptivity} for further discussion. The mean force ${\mb F}(\mb q^{\ast}) = -\nabla_{\mb q^{\ast}} U(\mb q)$ at a target point $\mb q^{\ast}$ can be estimated from the restrained dynamics by introducing a biased quadratic potential to the full MD system, i.e., 
\begin{equation}
\widetilde{V}(\mb Q; k, \mb q^{\ast}) = V(\mb Q) + V_b(\mb Q; k, \mb q^{\ast}) \quad V_b(\mb Q; k, \mb q^{\ast})=  \frac{k}{2}  \left\vert \phi^{q}(\mb Q) - \mb q^{\ast}\right\vert^2,
\label{eq:restrained_MD}
\end{equation}
where $k$ is the magnitude of the biased potential. $\mb F(\mb q^{\ast})$ can be sampled from the equilibrium distribution under biased potential $\rho_b(\mb Q; k, \mb q^{\ast}) \propto \exp(-\widetilde{V}(\mb Q; k, \mb q^{\ast}))$, i.e., 
\begin{equation}
\begin{split}
\mb F(\mb q^{\ast}) &= \lim_{k\to \infty} \int -\nabla_{\mb Q} V_b(\mb Q; k, \mb q^{\ast}) \rho_b(\mb Q; k, \mb q^{\ast}) \diff \mb Q \\
\lim_{k\to \infty} \int \rho_b(\mb Q; k,\mb q^{\ast}) \diff \mb Q &= \int \delta(\phi_q(\mb Q) - \mb q^{\ast}) \rho_0(\mb Q) \diff \mb Q :=\rho_q(\mb q^{\ast}),    
\end{split}
\label{eq:bias_dist}
\end{equation}
and we refer to Ref. \cite{TAMD_CPL_2006} for the proof. Accordingly, the sampling of $\rho_b(\mb Q; k, \mb q^{\ast})$ by Eq. \eqref{eq:restrained_MD} provides training data not only for constructing the free energy, but also for computing dynamical quantities. Specifically, for each $\mb q^{\ast}$, it also provides a natural approach to collecting the training samples conditional on the initial state $\phi_q(\mb Q) = \mb q^{\ast}$
\begin{equation}
S(\mb q^{\ast}) = \left\{\mb q^{(j,l)}  = \phi_q(\mb Q^{(l)}(t_j))  \left \vert \rho\left(\mb Q^{(l)}(0)\right) = \rho_b\left(\mb Q; \mb q^{\ast}\right) \right. \right\}_{j,l=1}^{N_j, N_l},
\label{eq:S_q_ast}
\end{equation}
which are used to estimate the mass matrix $\mb m(\mb q)$ via Eq.~\eqref{eq:mass} and the state-dependent correlations for learning the memory term $\mb K(\mb q, t)$ discussed in Section~\ref{sec:data_driven_learning}.  Therefore, the CAS framework establishes a shared sampling strategy which provides the training data for all the reduced modeling terms $U(\mb q)$, $\mb m(\mb q)$, and $\mb K(\mb q, t)$. This avoids the need to explicitly evaluate the full orthogonal dynamics ${\rm e}^{\mathcal{Q} \mathcal{L} t}$ and ensures consistent data collection across modeling terms.

To collect $\mb F(\mb q^{\ast})$ and $S(\mb q^{\ast})$ over the CG space, we need to establish efficient sampling of $\rho^{\ast}(\mb q)$ by Eq. \eqref{eq:max_problem}. However, the analytical solution is only formal since the residual $\mathcal{L}_{\mathcal{N}}(\mb q)$ is unknown \emph{a priori} and the numerical query at an arbitrary point can be computationally expensive or even thermodynamically inaccessible. As a result, common approaches such as Markovian Chain Monte Carlo and
Langevin-type algorithms can not be directly used. Instead, the CAS method introduces an interacting particle system $\left\{\mb q_i\right\}_{i=1}^{N_w}$ and adaptively seeks a local quadratic approximation of $\mathcal{L}_{\mathcal{N}}(\mb q)$ near the max-residual region, i.e., 
\begin{equation*}
G(\mb q; \mathcal{M}, \mathcal{V}) = \frac{1}{2}(\mb q - \mathcal{M})^{\top}\mathcal{V}^{-1}(\mb q-\mathcal{M}).      
\end{equation*}
Specifically, $\mathcal{M}$ and $\mathcal{V}$ represent the mean and covariance matrix of a weighted empirical distribution of the particle system, i.e.,
\begin{equation}
\begin{split}
\rho_{l}(\mb q) &= \sum_{i=1}^{N_w} \delta(\mb q - \mb q_i) \exp(-\kappa_l \mathcal{L}^{-}_{\mathcal{N}}(\mb q_i)) / \sum_{i=1}^{N_w} \exp(-\kappa_l \mathcal{L}^{-}_{\mathcal{N}}(\mb q_i)) \\
\mathcal{M}[\rho_{l}] &= \int \mb q  \rho_{l}(\mb q) \diff \mb q \\
\mathcal{V}[\rho_{l}] &= \kappa_t \int (\mb q - \mathcal{M}[\rho_{l}]) \otimes (\mb q - \mathcal{M}[\rho_{l}])   \rho_{l}(\mb q) \diff \mb q, \\
\end{split}
\label{eq:rho_M_V}
\end{equation}
where $\kappa_l$ and $\kappa_t$ are parameters interpreted as the inverse of the temperature. 
Following the Laplace's principle, 
by taking the low temperature limit ($\kappa_l \to \infty)$, $G(\mb q)$ provides a quadratic approximation of $\mathcal{L}^{-}_{\mathcal{N}}(\mb q)$ near the max-residual point. 

To sample $\rho^{\ast}(\mb q)$, we treat each particle as a random walker governed by the following McKean-type SDE
\begin{equation}
\label{eq:mcKean_SDE}
    \begin{split}
     \intd \bq_i(t) = \frac{1}{\gamma} \nabla_q G(\bq_i(t); \mathcal{M}_t , \mathcal{V}_t ) \intd t + \sqrt{\frac{2}{\kappa_h\gamma}}\intd \bW(t), \quad i = 1, \cdots, N_w, 
    \end{split}
\end{equation}
where $\mathcal{M}_t = \mathcal{M}[\rho_l^t]$, $\mathcal{V}_t = \mathcal{V}[\rho_l^t]$ and $\rho_l^t$ represents the instantaneous weighted empirical distribution of $\left\{\mb q_i(t)\right\}_{i=1}^{N_w}$. 
Consequently, the mean-field potential $G(\bq)$ exploits the local quadratic approximation of the residual $\mathcal{L}_{\mathcal{N}}$ under a low-temperature limit and drives the random walkers towards $\mathcal{M}_t$, which represents the maximal points conditioned on the information exploited so far. Meanwhile, the second stochastic term promotes the exploration of uncharted space, where $\gamma$ denotes the friction coefficient and $\bW(t)$ denotes the standard Brownian motion. 

Model \eqref{eq:mcKean_SDE} modulates the coupling between the exploitation and exploration (i.e., the conservative and stochastic term) via a high temperature $\kappa_h^{-1}$ consistent with the target distribution $\rho^{\ast}(\mb q)$ defined by Eq. \eqref{eq:max_problem}. Specifically,  we show that by choosing $\kappa_t= \kappa_h +\kappa_l$, the distribution will converge to the target one under appropriate conditions. 
\begin{proposition}
\label{prop:invariant_dist}
Suppose $\mathcal{L}_\mathcal{N}^-(\mathbf{q})$ takes a local quadratic approximation in form of 
$\frac{1}{2} (\mathbf{q}-\mathbf{\mu})^\top\Sigma^{-1}(\mathbf{q}-\mathbf{\mu})$. 
If the dynamics converge to an invariant distribution, then the stationary density is given by 
\begin{equation}
    \rho_\infty =  \frac{\exp{(-\kappa_h \mathcal{L}_\mathcal{N}^-(\mathbf{z}))}}{\displaystyle \int \exp {(-\kappa_h \mathcal{L}_\mathcal{N}^-(\mathbf{z}))} \mathrm{d} \mathbf{z}},
\end{equation} by choosing $\kappa_t = \kappa_l + \kappa_h$.
\end{proposition}
\begin{proof}
Let $\rho_{\infty}(\mathbf z)$ denote the invariant distribution of Eq. \eqref{eq:mcKean_SDE}. Then $\rho_{\infty} (\mathbf z)$ must be the invariant distribution of the following SDE 
\begin{equation}
\mathrm d \mathbf q = - \frac{1}{\gamma} \mathcal{V}^{-1}_{\kappa_l, \infty}(\mathbf q - \mathcal{M} _{\kappa_l, \infty})  \mathrm d t + \sqrt{\frac{2}{\gamma \kappa_h} } \mathrm d \mb W_t ,
\label{eq:asym_Langevin}
\end{equation}
where $\mathcal{M} _{\kappa_l, \infty}$ and $\kappa_t^{-1} \mathcal{V}_{\kappa_l, \infty}$ are the mean and the covariance matrix of the re-weighted density $\propto ~\rho_{\infty}(\mathbf z) e^{-\kappa_l \mathcal{L}_{\mathcal{N}}^{-}(\mathbf z)}$. With the fluctuation-dissipation relation for Eq. \eqref{eq:asym_Langevin}, we can show $\rho_{\infty} (\mathbf z)$ follows the Gaussian distribution with mean $\mathcal{M} _{\kappa_l, \infty}$ and covariance matrix $\kappa_h^{-1} \mathcal{V}_{\kappa_l, \infty}$. 

Since $\mathcal{L}_{\mathcal{N}}(\mathbf q) = \frac{1}{2} (\mathbf q-\mathbf\mu)^\top\Sigma^{-1}(\mathbf q-\mathbf\mu)$ is quadratic, the re-weighted density of a Gaussian distribution $\rho(\mathbf z) \sim \mathcal{N}(\mathcal{M} , \mathcal{V} )$ remains Gaussian, i.e., $
\rho(\mathbf z) e^{-\kappa_l \mathcal{L}_{\mathcal{N}}^{-}(\mathbf z)}  \propto ~  \mathcal{N} ({\mathcal{M} }_{\kappa_l}, {\mathcal{V} }_{\kappa_l}),
$ where ${\mathcal{M} }_{\kappa_l}$ and $\mathcal{V}_{\kappa_l}$ are defined by
\begin{equation}
\begin{split}
\mathcal{M}_{\kappa_l}(\mathcal{M}, \mathcal{V} ) &= (\mathcal{V} ^{-1} + \kappa_l \Sigma ^{-1}) ^{-1} (\kappa_l\Sigma^{-1}\mathbf\mu+V^{-1}\mathcal{M} ),\\
\mathcal{V}_{\kappa_l}(\mathcal{M}, V) &= (\mathcal{V} ^{-1} + \kappa_l \Sigma ^{-1}) ^{-1}.
\end{split}
\label{eq:m_v_reweight_Gauss}
\end{equation}
Therefore, the mean and covariance of the steady-state Gaussian distribution satisfy
\begin{equation}
\begin{aligned}
\mathcal{M} _{\kappa_l, \infty} &= (\kappa_h \mathcal{V} _{\kappa_l, \infty}^{-1} + \kappa_l \Sigma ^{-1})  (\kappa_l\Sigma^{-1}\mathbf\mu+\kappa_h \mathcal{V} _{\kappa_l, \infty}^{-1} \mathcal{M} _{\kappa_l, \infty}),\\ 
\mathcal{V} _{\kappa_l, \infty}  & = \kappa_t (\kappa_h \mathcal{V} _{\kappa_l, \infty}^{-1} + \kappa_l \Sigma ^{-1}) ^{-1}.
\end{aligned}
\nonumber
\end{equation}
It is easy to show that by choosing $\kappa_t = \kappa_l + \kappa_h$, $\mathcal{M} _{\kappa_l, \infty}$ and $\mathcal{V} _{\kappa_l, \infty}$ recovers $\mathbf\mu$ and $\Sigma$, respectively, and the invariant distribution takes the form
\begin{equation}
\rho_{\infty}(\mathbf z) \sim \mathcal{N}\left(\mathbf\mu, \kappa_h^{-1} \Sigma\right).
\nonumber
\end{equation}
\end{proof}

While the proof is only formal, it illustrates the essential idea of the present CAS method. A more rigorous version can be found in Proposition 2.5 of Ref. \cite{lyu2023consensus}. Specifically, we show that by choosing the parameter $\kappa_t = \kappa_l + \kappa_h$ for $\mathcal{V}$ in Eq. \eqref{eq:rho_M_V}, the empirical distribution of the random walker \eqref{eq:mcKean_SDE} converges to the target distribution $\rho^{\ast} (\mb q)$ with the first and second moments converging at an exponential rate under appropriate conditions. In particular, unlike common enhanced sampling methods, Eq. \eqref{eq:mcKean_SDE} only depends on the local residual and does not explicitly rely on the underlying atomistic potential function. This unique feature enables us to choose a larger time step and achieve efficient sampling for systems with complex energy landscapes. 

In order to ensure a stable estimation, we introduce a moving average to the computation of the mean and covariance
\[
\begin{aligned}
   \mathcal {M}_{t+1} = \beta_1 \mathcal {M}_{t} + (1-\beta_1) \mathcal M [\rho_l],\\
\mathcal {V}_{t+1} = \beta_2 \mathcal {V}_{t} + (1-\beta_2) \mathcal V [\rho_l]. 
\end{aligned}
\]
A step-dependent normalizer is also introduced to ensure this estimation is unbiased. 
Algorithm. \ref{alg:cas} summarizes the detailed sampling process.

\begin{remark}
\label{rem:adaptivity}
In principle, the memory term $\mb K(\mb q, t)$   can be directly evaluated by approximating the orthogonal dynamics ${\rm e}^{\mathcal{Q} \mathcal{L} t}$ in Eq. \eqref{eq:memory} (e.g., see Refs. \cite{hijon2010mori, lyu2023construction}) where the residual between $\mb K(\mb q, t)$ and its surrogate can be incorporated into the adaptive metric $X(\mb q)$. 
In this work, we only impose the sampling adaptivity based on the free energy by choosing $X(\mb q) = -\nabla U(\mb q)$. 
This choice ensures that samples are concentrated in kinetically important regions while simultaneously enabling the collection of dynamical observables for $\mb m(\mb q)$ and $\mb K(\mb q,t)$. This shared sampling strategy enables us to avoid dealing with the full-dimensional orthogonal dynamics ${\rm e}^{\mathcal{Q} \mathcal{L} t}$ and establish an efficient data-driven construction of $\mb K(\mb q, t)$ presented in Sec. \ref{sec:data_driven_learning}. 
\end{remark}

\begin{algorithm}
\caption{Consensus-based enhanced sampling.}
\begin{algorithmic}
\REQUIRE{Initial sampling point $\mathbf{q}_{i,0}$, for $i=1,\ldots,N_{w}$}
\REQUIRE{Initial NN parameter $\theta_0$}
\REQUIRE{The number of training iterations $N_{train}$}
\REQUIRE The number of data collected $N_{data}$ in each training iteration
\STATE{$j \gets 0, t \gets 0 $}
\STATE $T \gets \lceil \frac{N_{data}}{N_{w}} \rceil $
\WHILE{$j < N_{train}$}
\WHILE{$t \leq T$}
\STATE calculate the mean force $\mathbf{F}(\bq_{i,t})$ and the training samples $S(\bq_{i,t})$ 
\STATE calculate the predicted force $\mathcal{F}_\theta(\mathbf{q}_{i,t}) = \nabla_\mathbf{q} X_{\mathcal{N}}(\mathbf{q}_{i,t}; \theta_j)$
\STATE $L^i \gets \mathcal{L}_\mathcal{N}(\mathbf{q}_{i,t})$
\STATE $w^i \gets \frac{\exp{(\kappa_l L^i})}{\sum_i \exp{(\kappa_l L^i})}$
\STATE $\mathcal{M}_{t+1} \gets \beta_1 \mathcal{M}_{t} + (1-\beta_1) \sum_i\mathbf{q}_{i,t} w^i$ 
\STATE $\mathcal{V}_{t+1} \gets \beta_2 \mathcal{V}_{t}+ (\kappa_l+\kappa_h)(1-\beta_2)\sum_i(\mathbf{q}_{i,t}-\mathbf{m}_t)^2 w^i$
\STATE $\mathcal{M} \gets \frac{\mathcal{M}_{t+1}}{1-\beta_1^t}$
\STATE $\mathcal{V} \gets \frac{\mathcal{V}_{t+1}}{1-\beta_2^t}$
\STATE $\mathbf{q}_{i,t+1} \gets \mathbf{q}_{i,t} - \frac{\delta t}{\gamma} (\mathbf{q}_{i,t}-\mathcal{M})\odiv\mathcal{V}+ \sqrt{\frac{2\delta t}{\gamma\kappa_h}}\mathbf{\eta}_{i,t}, \eta_{i,t} \sim \mathcal{N}(0,1)$
\STATE $t \gets t+1 $
\ENDWHILE
\STATE Save the training dataset $\mathcal{D}_j=\{\mathbf{q}_{i,t},\mathbf{F}(\bq_{i,t}),S(\bq_{i,t})\}_{t=0}^T$
\STATE Optimize $\theta_{j+1}$ using the generated training set $\mathcal{D}_l$ for $l=0,\ldots,j$.
\STATE $j \gets j+1 $
\ENDWHILE
\end{algorithmic}
\label{alg:cas}
\end{algorithm}

\subsection{Data-driven learning of the state-dependent memory}
\label{sec:data_driven_learning}

Let $\mathcal{S} = \left\{\mb q^{(k)}, \mb F\left (\mb q^{(k)}\right), S(\mb q^{(k)}) \right\}_{k=1}^{N_s}$ and $U(\mb q)$ denote the training set and the learned free energy function obtained from the CAS-based sampling process. The surrogate for the mass matrix $\mb m(\mb q)$ is constructed by minimizing the empirical discrepancy between the exact definition \eqref{eq:mass} and the velocity correlation function
\begin{equation}
\mathcal{L}_m = \sum_{k=1}^{N_s} \left \Vert \mb m(\mb q^{(k)}) - \beta^{-1} \mb C_{vv}(0, \mb q^{(k)}) \right \Vert^2,     
\end{equation}
where $\mb C_{vv}(t, \mb q^{(k)})$ is estimated from the sample set $S(\mb q^{(k)})$ introduced in Eq.~\eqref{eq:S_q_ast}.
In this subsection, we focus on the approximation of 
the memory term $\mb K(\mb q, t)$. 

To capture the state-dependent non-Markovian nature, we embed $\mb K(\mb q, t)$ into an extended Markovian system in terms of $\widehat{\mb z} = \left[\mb q; \mb p; \bm \chi\right]$, where $\bm \chi = \left [\bm \chi_1; \cdots; \bm \chi_n\right]$ represents a set of non-Markovian features and the individual feature $\bm \chi_i \in \mathbb{R}^{m}$ will be learned from the full model as specified later. The extended dynamics takes the general form 
\begin{equation}
{\rm d}
\begin{pmatrix} \mb q \\ \mb p \\ \bm \chi
\end{pmatrix} = \begin{pmatrix}
0 & \begin{matrix}\mb I & ~0 \end{matrix} \\
~\begin{matrix} -\mb I \\ 0 \end{matrix} & \mb J(\mb q)
\end{pmatrix} \nabla \widehat{F}(\mb q, \mb p, \bm \chi) {\rm d}t 
+ \begin{pmatrix}
0 & \begin{matrix} 0 & ~0 \end{matrix} \\
~\begin{matrix} 0 \\ 0 \end{matrix} & \mb \Sigma(\mb q)
\end{pmatrix}  {\rm d} \mb W_t,
\label{eq:extended_dynamics}
\end{equation}
where $\widehat{F}(\mb q, \mb p, \bm \chi)$ is the free energy function of the extended system defined by 
\begin{equation}
\widehat{F}(\mb q, \mb p, \bm \chi) = U(\mb q) + \frac{1}{2} \mb p^\top \mb m(\mb q)^{-1} \mb p  + \frac{1}{2} \beta^{-1} \ln {\det \left[\mb m(\mb q)\right]} + \frac{1}{2} \bm \chi^\top \bm \chi.
\label{eq:extend_free_energy}
\end{equation}
The matrices $\mb J(\mb q)$ and $\mb \Sigma(\mb q)$ govern the energy dissipation and the noise, and are constructed as
\begin{equation}
\mb J(\mb q) = \begin{pmatrix}
0 & -\mb H(\mb q) \\ 
\mb H(\mb q)^\top & \mb \Gamma (\mb q)
\end{pmatrix} \quad \mb \Sigma(\mb q) = \begin{pmatrix}
0 &0\\
0 &\mb D(\mb q)
\end{pmatrix},
\label{eq:J_Sigma}
\end{equation}
where $\mb H:\mathbb{R}^{m}\to\mathbb{R}^{n\times nm}$ encodes the coupling between the CG momentum and the auxiliary variables and $\bm\Lambda, \mb D:\mathbb{R}^{m}\to\mathbb{R}^{nm\times nm}$ characterize the dissipation and noise, respectively.
In particular, with the proper choice of the multiplicative noise, we can show that the reduced model \eqref{eq:extended_dynamics} retains a consistent invariant distribution function of the full MD system. 
\begin{proposition}
By choosing $\mb D(\mb q)\mb D(\mb q)^\top = - \beta^{-1}\left(\bm \Gamma (\mb q) + \bm \Gamma (\mb q)^\top\right)$, Eq. \eqref{eq:extended_dynamics} has an invariant distribution 
\begin{equation}
\rho_{\rm 0}(\mb q, \mb p, \bm \chi) \propto \frac{1}{(2\pi)^{m/2} \det [\mb m(\mb q)]^{1/2}}\exp\left[-\beta \left(U(\mb q) + \frac{1}{2} \mb p^\top \mb m(\mb q)^{-1} \mb p  + \frac{1}{2} \bm \chi^\top \bm \chi\right)\right]
\label{eq:extend_invariant}
\end{equation}
\label{prop:extended_dynamics}
\end{proposition}
\begin{proof}
For simple notation, we write Eq. \eqref{eq:extended_dynamics} as 
\begin{equation*}
{\rm d} \widehat{\mb z} = \widehat{\mb J}(\mb q) \nabla \widehat{F}(\widehat{\mb z}) {\rm d}t + \widehat{\bm \Sigma}(\mb q) {\rm d}\mb W_t.     
\end{equation*}
The Fokker-Planck equation is given by  
\begin{equation*}
\partial_t\rho = \nabla \cdot \left ( - \widehat{\mb J}(\mb q) \nabla \widehat{F}(\widehat{\mb z}) \rho + \frac{1}{2}\nabla \cdot\left(\widehat{\mb \Sigma} (\mb q)  \widehat{\mb \Sigma} (\mb q)^\top \rho \right ) \right). 
\end{equation*}
With $\mb J(\mb q)$ and $\bm \Sigma(\mb q)$ taking the form \eqref{eq:J_Sigma}, we can show that $\nabla \cdot \left(\left(\mb J(\mb q) - \mb J(\mb q)^\top\right) \nabla g(\widehat{\mb z})\right) \equiv 0 $ for any scalar function $g$ and $\nabla \cdot\left(\widehat{\mb \Sigma} (\mb q)  \widehat{\mb \Sigma} (\mb q)^\top \rho\right ) = \widehat{\mb \Sigma} (\mb q)  \widehat{\mb \Sigma} (\mb q)^\top \nabla \rho$. By plugging $\rho_0$ into the right-hand side, we have 
\begin{equation*}
\begin{split}
\nabla \cdot \left ( - \widehat{\mb J}(\mb q) \nabla \widehat{F}(\widehat{\mb z}) \rho_0 + \frac{1}{2}\nabla \cdot\left(\widehat{\mb \Sigma} (\mb q)  \widehat{\mb \Sigma} (\mb q)^\top \rho_0 \right ) \right) &=  \nabla \cdot \left ( \beta^{-1} \widehat{\mb J}(\mb q) \nabla \rho_0 + \frac{1}{2} \widehat{\mb \Sigma} (\mb q)  \widehat{\mb \Sigma} (\mb q)^\top \nabla \rho_0 \right) \\
&= \nabla \cdot \left ( \frac{\widehat{\mb J}(\mb q) + \widehat{\mb J}(\mb q)^\top}{2\beta} \nabla \rho_0 + \frac{1}{2} \widehat{\mb \Sigma} (\mb q)  \widehat{\mb \Sigma} (\mb q)^\top \nabla \rho_0 \right) \\
&\equiv 0 
\end{split}    
\end{equation*}
\end{proof}

By proposition \ref{prop:extended_dynamics}, it is natural to construct $\bm\Gamma(\mb q)$ by
\begin{equation*}
\mb \Gamma(\mb q) = -\mb L(\mb q) \mb L(\mb q)^{\top} + \mb J^a(\mb q), \quad \mb D(\mb q) = \beta^{-1/2}\mb L(\mb q)
\end{equation*}
where $\mb L(\mb q)$ is a lower-triangular matrix and $\mb J^{a}(\mb a)$ is a skew-symmetric matrix represented by neural networks.  As a special case, by taking $\mb H(\mb q)$ and  $\bm\Lambda(\mb q)$ as constant matrices, the reduced model \eqref{eq:extended_dynamics} recovers the standard GLE with a homogeneous kernel \cite{she2023data}. In this study, we do not make such a simplification; the embedded memory enables us to capture the heterogeneous energy dissipation process overlooked in the standard GLE.  

To train the reduced model, we establish a joint learning of both the auxiliary variables $\left\{\bm\chi_i\right\}_{i=1}^n$ and the functions $\mb H(\mb q)$ and $\bm\Lambda(\mb q)$. Specifically, $\bm\chi_i (t)$ essentially serves as a set of non-Markovian features that encode the unresolved dynamics orthogonal to $\mb z(t) = \mathcal{P}\mb Z(t)$. However, the direct construction in terms of $\mathcal{Q}\mb Z(t)$ can be computationally expensive due to the high-dimensionality of $\mb Z(t)$. Alternatively, one important observation is that the time history of $\mb p(t)$ naturally encodes the unresolved dynamics and can be readily obtained from the training set $\mathcal{S}$ discussed in Sec. \ref{sec:CAS}. Hence, we construct $\bm\chi_i (t)$ by 
\begin{equation}
\begin{split}
\bm \chi_i(t)  &= \int_0^{\infty} \bm \omega_i(t) \mb p(t-\tau) \diff \tau  \\
&\approx \sum_{j=1}^{N_w} \mb w_{ij} \mb p(t-t_j), \quad i = 1, \cdots, n,
\end{split}
\label{eq:chi_encoder}
\end{equation}
where $\bm\omega_i \in \mathbb{R}^{+} \to \mathbb{R}^{m\times m}$ is the non-Markovian encoders and $\left\{\mb w_{ij}\right\}_{j=1}^{N_w}$ represents the discrete weights whose values will be determined. 

To proceed, we multiply Eq. \eqref{eq:extended_dynamics} by $\mb v(0) := \dot{\mb q}(0)$ and take the conditional expectation with $\mb q(0) = \mb q^{\ast}$. The correlation functions follow 
\begin{equation}
\begin{split}
\frac{\diff }{\diff t}\left\langle \widehat{\mb z}(t) \mb v(0)^{\top} \vert \mb q^{\ast}\right\rangle &= \left\langle \mb J(\mb q(t)) \nabla \widehat{F}(\widehat{\mb z}(t)) \mb v(0)^\top \big \vert \mb q^{\ast} \right\rangle \\
&\approx \mb J(\mb q^{\ast}) \left\langle  \nabla \widehat{F}(\widehat{\mb z}(t)) \mb v(0)^\top \big \vert \mb q^{\ast} \right\rangle
\end{split}
\label{eq:J_separation}
\end{equation}
for $\widehat{\mb z}$ taking $\mb p$ and $\bm \chi$. In particular, the approximation of the second equation is based on the assumption that the position correlation function $\mb C_{qq}(t)$ decays much slower than the velocity correlation function $\mb C_{vv}(t)$ for the CG variables. The separation enables us to efficiently evaluate Eq. \eqref{eq:J_separation} with pre-computed correlation functions
rather than the on-the-fly computation from the time-series samples of $\mb q$ and $\mb p$; see Fig. \ref{fig:2D_fes_corr} for the numerical verification of this assumption. By using Eq. \eqref{eq:J_Sigma}, we have
\begin{equation}
\begin{split}
\frac{\diff}{\diff t}
\underbrace{
\begin{pmatrix}
{\left\langle \mb p (t) \mb v(0)^\top \vert \mb q^{\ast}\right\rangle} \\
{\left\langle \bm\chi_1 (t)  \mb v(0)^\top \vert \mb q^{\ast}\right\rangle} \\
\vdots \\
\left\langle {\bm\chi}_{n} (t)  \mb v(0)^\top \vert  \mb q^{\ast} \right\rangle
\end{pmatrix}
}_{\mb C_1(t,{\mb q}^{\ast})}
+
\underbrace{
\begin{pmatrix}
{\left\langle \nabla_{\mb q(t)} F(\mb q, \mb p) \mb v(0)^{\top} \vert \mb q^{\ast}\right\rangle} \\
0 \\
\vdots \\
0
\end{pmatrix}
}_{\mb C_0(t,{\mb q}^{\ast})}
&= \mb J(\mb q^{\ast}) \underbrace{ 
\begin{pmatrix}
{\left\langle \mb{v} (t) \mb v(0)^{\top} \vert \mb q^{\ast} \right\rangle} \\
{\left\langle \bm\chi_1 (t) \mb v(0)^{\top} \vert \mb q^{\ast} \right\rangle} \\
\vdots \\
{\left\langle  \bm\chi_{n}(t) \mb v(0)^{\top} \vert \mb q^{\ast} \right\rangle}
\end{pmatrix}  
}_{\mb C_2(t, {\mb q}^{\ast})},
\label{eq:correlation_evolution}
\end{split}
\end{equation}
where the correlation $\left\langle \bm\chi_i(t)\mb v(0)^\top \vert \mb q^{\ast}\right\rangle$ can be obtained from the non-Markovian weights $\mb w_{ij}$ and $\mb C_{pv}(t; \mb q^{\ast})$ by Eq. \eqref{eq:chi_encoder}.  Accordingly, the correlation matrices $\mb C_0(t, {\mb q}^{\ast})$, $\mb C_1(t, {\mb q}^{\ast})$ and $\mb C_2(t, {\mb q}^{\ast})$ can be directly evaluated for each $\mb q^{\ast}$ in the training set. This enables us to establish a joint learning of both the non-Markovian weights $\left\{\mb w_{ij}\right\}_{i,j=1}^{n,N_w}$ and the function $\mb J(\mb q)$ (i.e., $\mb H(\mb q)$, $\mb L(\mb q)$ and $\mb J^{a}(\mb q)$) by minimizing the empirical loss function 
\begin{equation*}
\begin{split}
L_k &= \lambda_{c} L_{c} +  \lambda_{\Lambda} L_{\Lambda} \\
L_c &= \sum_{k=1}^{N_s} \sum_{j=1}^{N_t}  \left\Vert \frac{\diff}{\diff t} {\mb C_1}(t_j, \mb q^{(k)}) + \mb C_0(t_j, \mb q^{(k)}) -
\mb J (\mb q^{(k)}) {\mb C_2}(t_j, \mb q^{(k)})\right\Vert^2 \\
L_{\Lambda} &=  \sum_{k=1}^{N_s} \left \Vert \bm \Lambda (\mb q^{(k)}) - \mb I \right \Vert^2  \quad \bm \Lambda(\mb q^{\ast}) = \left\langle \bm \chi(t) \bm \chi(t)^\top \vert \mb q(t) = \mb q^{\ast} \right\rangle,
\end{split}
\end{equation*}
where $\lambda_c$ and $\lambda_{\Lambda}$ are hyperparameters. The second loss term imposes a constraint such that the covariance matrix of the auxiliary variables $\bm\chi$ constructed from the full model is close to $\mb I$, which is consistent with the pre-assigned free energy function $\widehat{F}(\mb q, \mb p, \bm \chi)$ and the invariant distribution $\rho_0(\mb q, \mb p, \bm \chi)$ of the reduced model specified in Eqs. \eqref{eq:extend_free_energy} and \eqref{eq:extend_invariant}. We use the ADAM optimizer \cite{Kingma_Ba_Adam_2015} to train the model and refer to \ref{app:training} for the details.

\section{Numerical Result}
In this section, we examine the effectiveness of the constructed reduced models in comparison with the full MD model. We use SD-GLE to denote the present state-dependent GLE model and SI-GLE to denote the standard state-independent GLE. For fair comparison, both reduced models use the same free energy function $U(\mb Q)$.   

\subsection{Polymer Chain}
Let us start with a polymer chain consisting of $N=16$ atoms. The full atomistic potential function is governed by standard intramolecular interactions, including non-bonded Lennard-Jones, bond stretching, angular bending, and dihedral torsion potential functions; see  \ref{app:polymer} for details. We choose the CG coordinate as the end-to-end distance defined by the Euclidean distance between the first and last atoms of the polymer chain $q = \|\bQ_1-\bQ_N\|$. Accordingly, the mass $\mb m$ is independent of $q$ since $\left\langle \dot{q}\dot{q} \right\rangle = \beta^{-1}(M_1^{-1} + M_{N}^{-1})$. The training set consists of 25 distinct configurations sampled from the atomistic simulations. The number of auxiliary variables $\left\{\bm\chi_i\right\}_{i=1}^n$is set to be $n=4$ for both the SI-GLE and SD-GLE models. 

First, we examine the overall velocity autocorrelation function defined as $C_{vv}(t) = \left\langle \dot q(t)\dot q(0)\right\rangle$. 
As shown in Fig. \ref{fig:1D_molecule_system_auto},  the predictions from both reduced models show good agreement with the full MD results, with the present SD-GLE model yielding improved accuracy in the intermediate regime for $t \sim O(1)$. This result is not surprising since $C_{vv}(t)$ represents a statistical property taking the ensemble average over the full CG space and can be well captured by the standard GLE with a homogeneous memory function. 

On the other hand, the differences between the two reduced models become more pronounced by examining the conditional velocity correlation $C_{vv}(t, q^{\ast}) = \left\langle \dot q(t)\dot q(0) \vert q(0) = q^{\ast} \right\rangle$. Fig. \ref{fig:1D_molecule_system_sd_auto} shows the obtained $C_{vv}(t, q^{\ast})$ for representative values $q^{\ast}$ ranging from $3$ to $16$. The apparent dispersive behavior indicates pronounced heterogeneous energy dissipation manifested by the state-dependent memory. The standard SI-GLE is insufficient to capture such behavior. In contrast, the present SD-GLE shows good agreement with the full MD results. 

\begin{figure}
    \centering
    \includegraphics[width=0.39\linewidth]{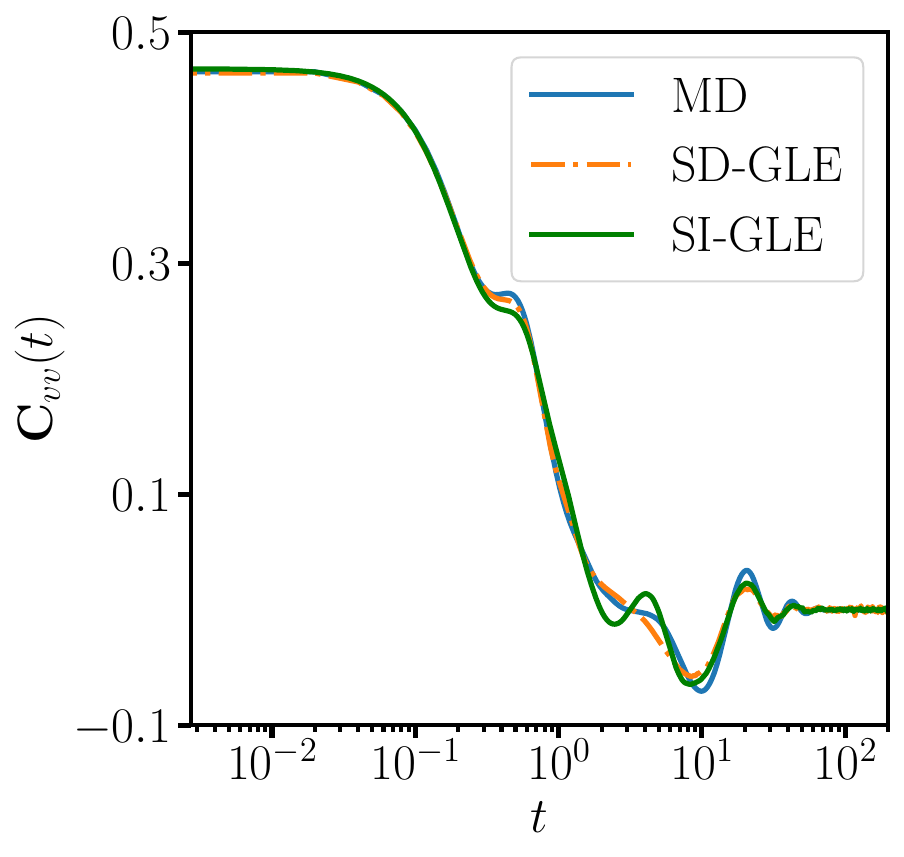}
    \caption{Velocity autocorrelation function $C_{vv}(t)$ obtained from the full MD, the present model (SD-GLE) and the standard GLE (SI-GLE), where the CG coordinate $q$ is defined as the end-to-end distance of a polymer molecule.}
    \label{fig:1D_molecule_system_auto}
\end{figure}

Finally, we study the rare event behavior for the large extension of the polymer configuration. In particular, we use $\rho(T\vert q> 15)$ to denote the distribution of the time duration of the polymer with $q > 15$. As shown in Fig. \ref{fig:1D_molecule_system_rare_event}, the prediction of the present SD-GLE model shows good agreement with the full MD results. In contrast, the prediction of the standard SI-GLE shows apparent deviations, which indicates the crucial role of the state-dependent memory function for capturing the conformation dynamics of the molecular system on the collective scale.

\begin{figure}
    \centering
    \includegraphics[width=0.95\linewidth]{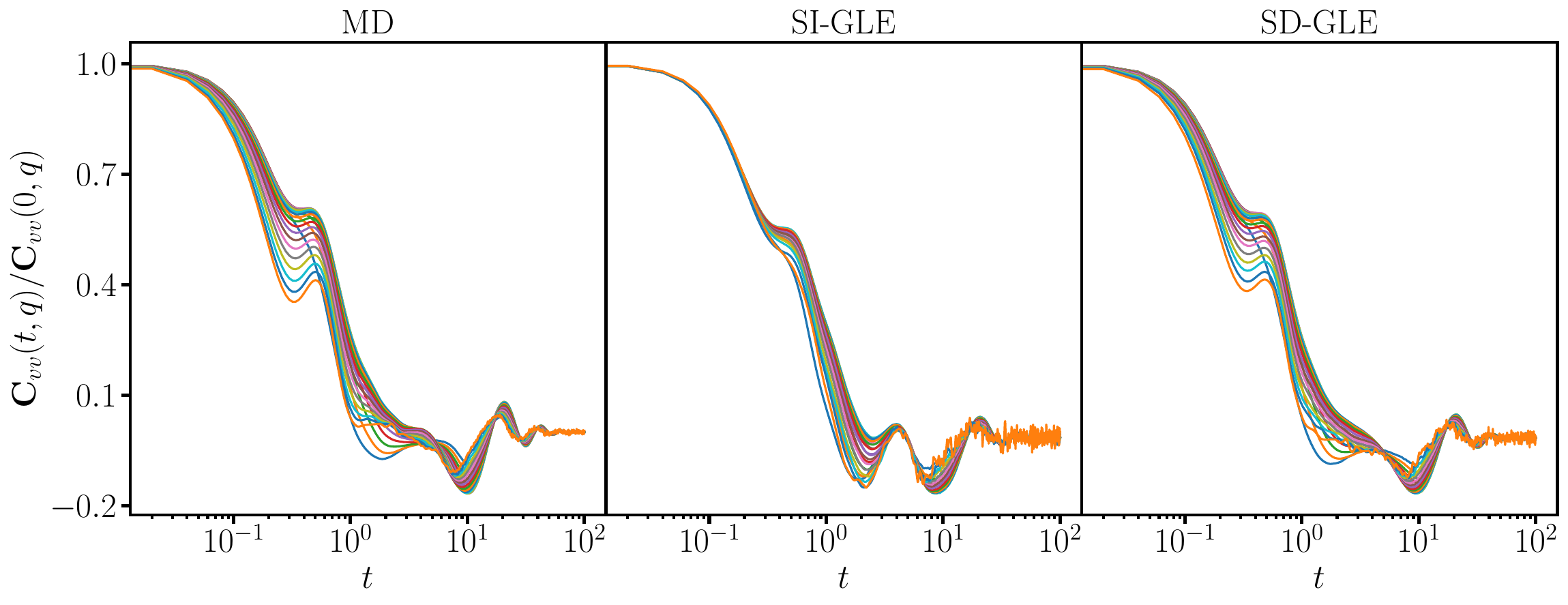}
    \caption{Conditional velocity autocorrelation function $C_{vv}(t, q^{\ast})$ of a polymer molecule with the initial state of the end-to-end distance $q(0) = q^{\ast}$. Each curve corresponds to the correlation function taking a specific value ranging between $3$ and $16$. The dispersive behavior indicates the pronounced state-dependent memory effect over the CG space.}
    \label{fig:1D_molecule_system_sd_auto}
\end{figure}

\begin{figure}
    \centering
    \includegraphics[width=0.36\linewidth]{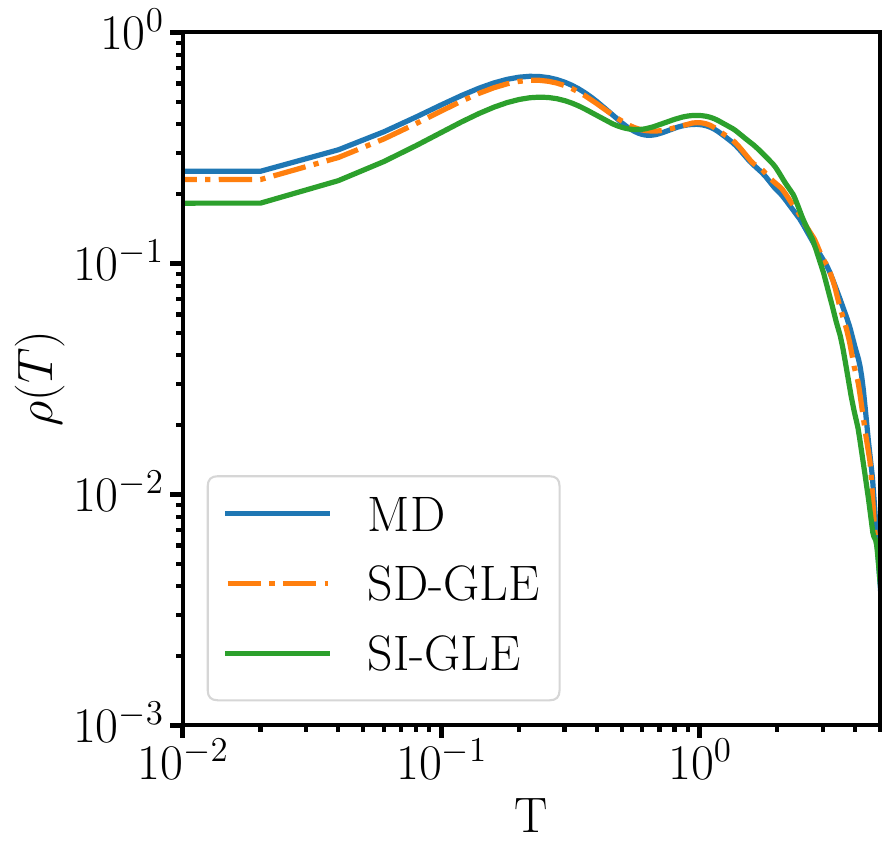}
    \caption{The rare-event distribution of the time duration for the polymer molecule under a large extension with the end-to-end distance $q>15$.}
    \label{fig:1D_molecule_system_rare_event}
\end{figure}

\subsection{Alanine Dipeptide}

Next, we consider an alanine dipeptide molecule solvated in $384$ water molecules at $300$ K. The AMBER99SB force field \cite{hornak2006comparison} is used to model the intramolecular and intermolecular interactions, and the TIP3P water model is used to model the solvent; we refer to \ref{app:ala_mol} for the details of the MD setup.

We choose the CG coordinates $\mb q = [\phi, \psi]$, where $\phi$ and $\psi$ represent the two backbone torsion angles. We use $10$ random walkers to sample the configuration space. The initial positions of the walkers at the $k\mhyphen$th iteration are designated as the terminal states of the preceding iteration. The inverse low and high temperatures are set to $\kappa_l=10$ and $\kappa_h=1$, respectively. For each sample point, we conduct $2$ ps of restrained dynamics to calculate the average force and mass, followed by $4$ ps of unrestrained dynamics to determine the conditional correlation functions. We use $n=18$ auxiliary variables to construct both the SI-GLE and SD-GLE models.

Fig. \ref{fig:2D_fes_corr} shows the constructed free energy function $U(\mb q)$ and the overall correlation functions $C_{qq}(t)$ and $C_{vv}(t)$. The four red points represent the individual metastable states separated by energy barriers where the transition among these states will be systematically investigated. In addition, the position correlation function $C_{qq}(t)$ decays much slower than the velocity correlation $C_{vv}(t)$, which verifies the separation of $\mb J(\mb q)$ from the conditional correlation function in Eq. \eqref{eq:J_separation}.  Fig. \ref{fig:2D_mass} shows the components of the mass matrix $\mb m(\mb q)$. Unlike the above polymer system, $\mb m(\mb q)$ exhibits a non-negligible state-dependent nature where the off-diagonal component $\mb m_{\phi, \psi}$ is about $10\%$ of the diagonal components. 

To further examine the localized memory effects, Figs. \ref{fig:2D_corr_123}-\ref{fig:2D_corr_999} shows the obtained conditional velocity correlation functions $C_{vv}(t, \mb q^{\ast})$ with $\mb q^{\ast}$ taking various metastable points shown in Fig. \ref{fig:2D_fes_corr}. For all the cases, the predictions from the present SD-GLE model show better agreement with the full MD results than the standard SI-GLE model. The improvement becomes even more pronounced for the conditional position correlation functions $C_{qq}(t, \mb q^{\ast})$, as shown in Figs. \ref{fig:2D_corr_qq_123}-\ref{fig:2D_corr_qq_999}. 


Finally, we examine the transition dynamics among the individual metastable states by considering the transition times between metastable states defined by 
\begin{equation}
\begin{split}
T_{A\to B} &= \inf \left\{t > 0 \,\middle|\, \mb q(t) \in B,\ \mb q(s) \notin A,~ \forall s \in (0, t) \right\}, \quad  \mb q(0) \in A. 
   \\
\end{split}
\label{eq:transition_time}
\end{equation}
Essentially, $T_{A\to B}$ is the committor-based transition time \cite{E_Vanden_ARPC_2010} that measures the first arrival at 
B after the system has last exited A, without re-entering A along the way. It characterizes the duration of the actual transition event without the recrossings.  
Fig. \ref{fig:2D_transition} shows the distribution of the transition time obtained from the full MD and the two reduced models. 
The SI-GLE model with a spatially homogeneous memory is insufficient to account for the local dynamical environments across the CG space. Consequently, it cannot capture heterogeneous dissipation and memory effects that are crucial for accurately resolving metastable transitions. 
In contrast, the present SD-GLE model shows consistent improvement in capturing these dynamical behaviors, which demonstrates the crucial role of the state-dependent memory term in predicting the non-equilibrium properties of molecular systems at the CG level.

\section{Summary}

We have presented a data-driven framework for constructing stochastic reduced models with state-dependent memory that accurately capture both equilibrium and non-equilibrium dynamics of molecular systems. The resulting SD-GLE model extends the standard GLE formulation by incorporating a memory kernel that varies across the CG space, thereby accounting for heterogeneous energy dissipation and complex metastable transitions.
To address the dual challenges of efficient sampling and non-Markovian modeling, we adopt a consensus-based sampling method that enables the efficient exploration of the phase space and the concurrent collection of free energy gradients, mass matrices, and conditional correlation functions through restrained dynamics. This shared-sampling approach makes the construction of the SD-GLE model a natural extension of existing GLE-based workflows. To approximate the state-dependent memory, we formulate an extended Markovian representation by embedding a set of auxiliary variables into the reduced dynamics. The memory term is learned jointly with these non-Markovian features through a conditional correlation matching scheme that requires only two-point statistics.

We demonstrate the accuracy and effectiveness of the proposed SD-GLE model on both a polymer chain and an alanine dipeptide molecule in solvent. The numerical results show clear improvements over standard GLE models in predicting the non-equilibrium properties, especially for the rare events and the transition dynamics. The present framework naturally extends existing workflows based on free energy modeling and provides a systematic framework towards data-driven modeling of state-dependent memory effects in complex molecular systems. Future directions include generalizing the framework to higher-dimensional CG representations and extending the methodology to non-equilibrium systems under external force fields \cite{Lyu_Lei_SIAM_MMS_2025}. Also, it could be interesting to incorporate the approximation residual of the memory term into the adaptive sampling process. We leave these for future studies.

\begin{figure}
    \centering
    \includegraphics[width=0.385
\linewidth]{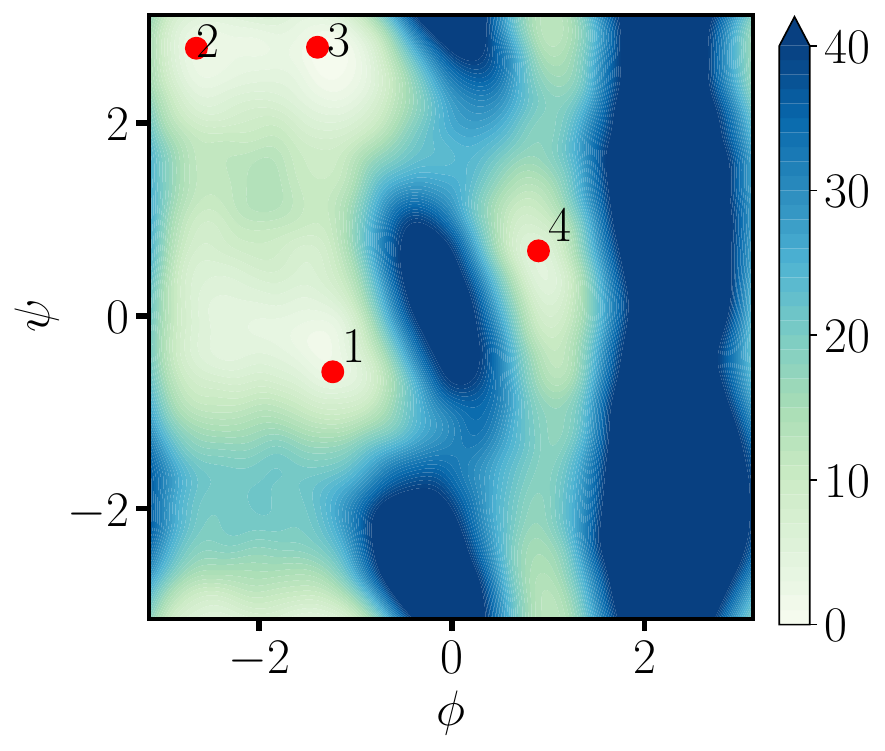}
    \includegraphics[width=0.33\linewidth]{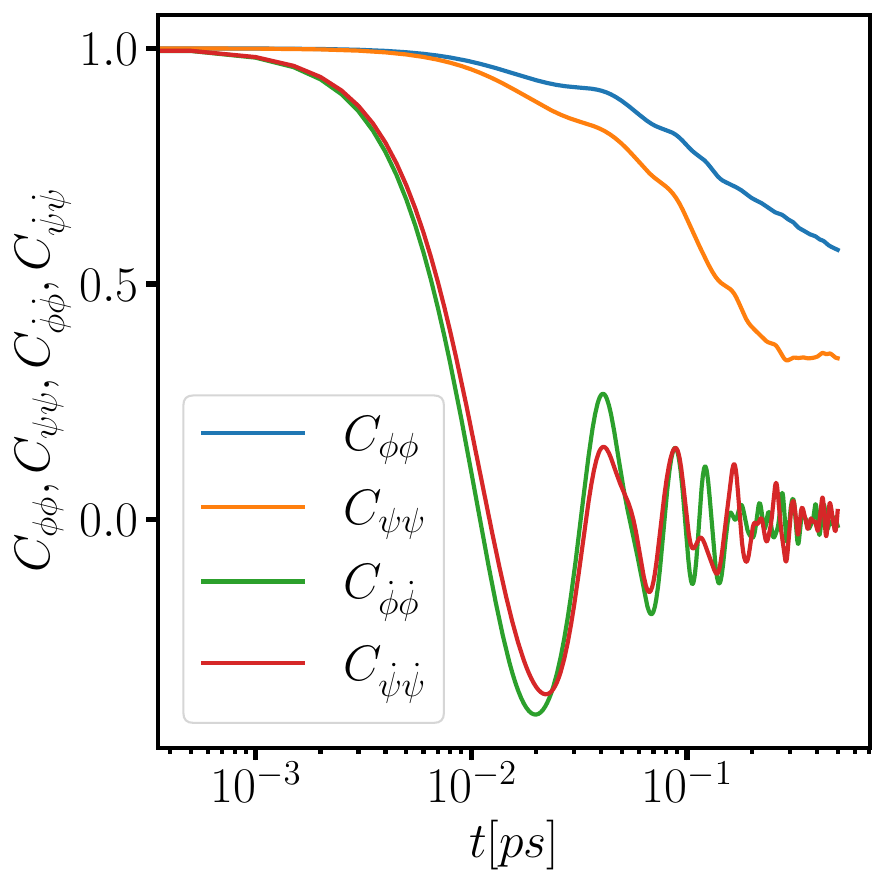}
    \caption{Left: The 2D free energy function $U(\mb q)$ for the reduced model of an alanine dipeptide molecule in solvent, where the CG coordinates $\mb q = \left[\phi, \psi\right]$ denote the two backbone torsion angles of the molecule. The four red points denote the individual metastable states separated by energy barriers. Right: Overall position correlation function $C_{qq}(t)$ and the velocity correlation function $C_{vv}(t)$ normalized by the initial values.}
    \label{fig:2D_fes_corr}
\end{figure}

\begin{figure}
    \centering
    \includegraphics[width=0.9\linewidth]{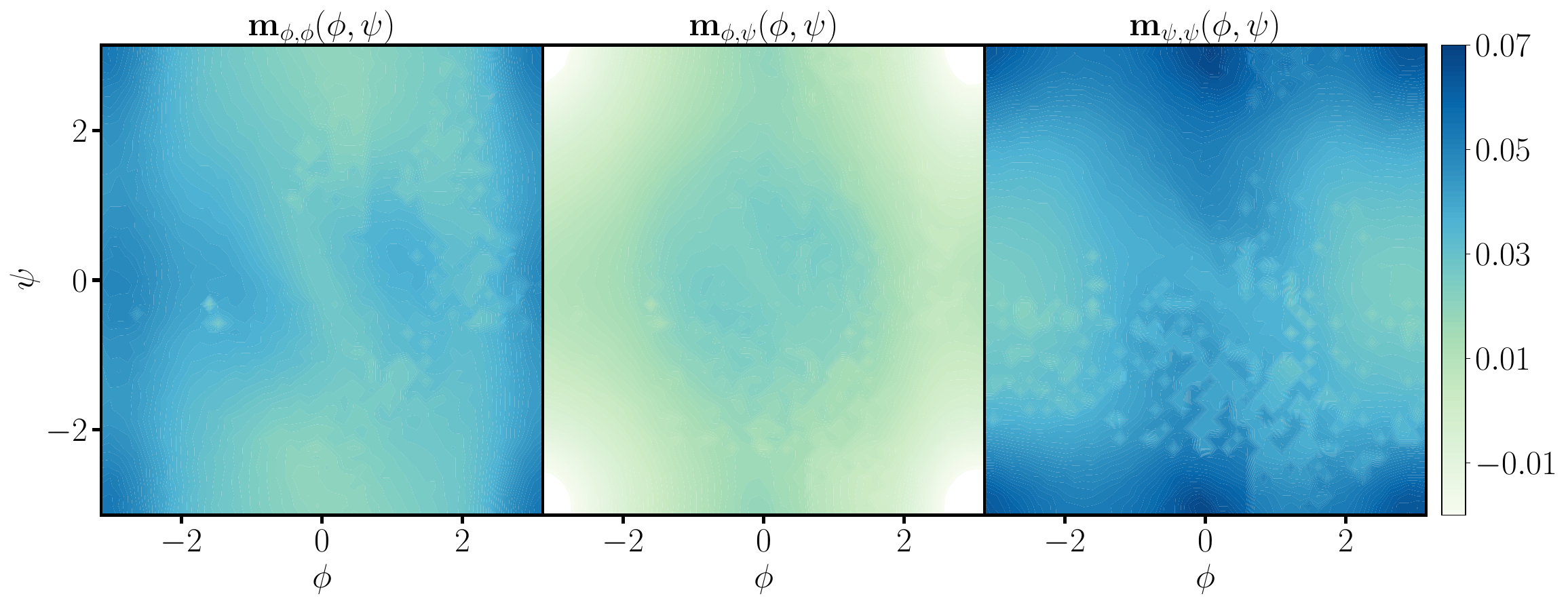}
    \caption{The 2D state-dependent mass matrix $\mb m (\mb q)$ for the reduced model of an alanine dipeptide molecule in solvent.}
    \label{fig:2D_mass}
\end{figure}

\begin{figure}
    \centering
    \includegraphics[width=0.85\linewidth]{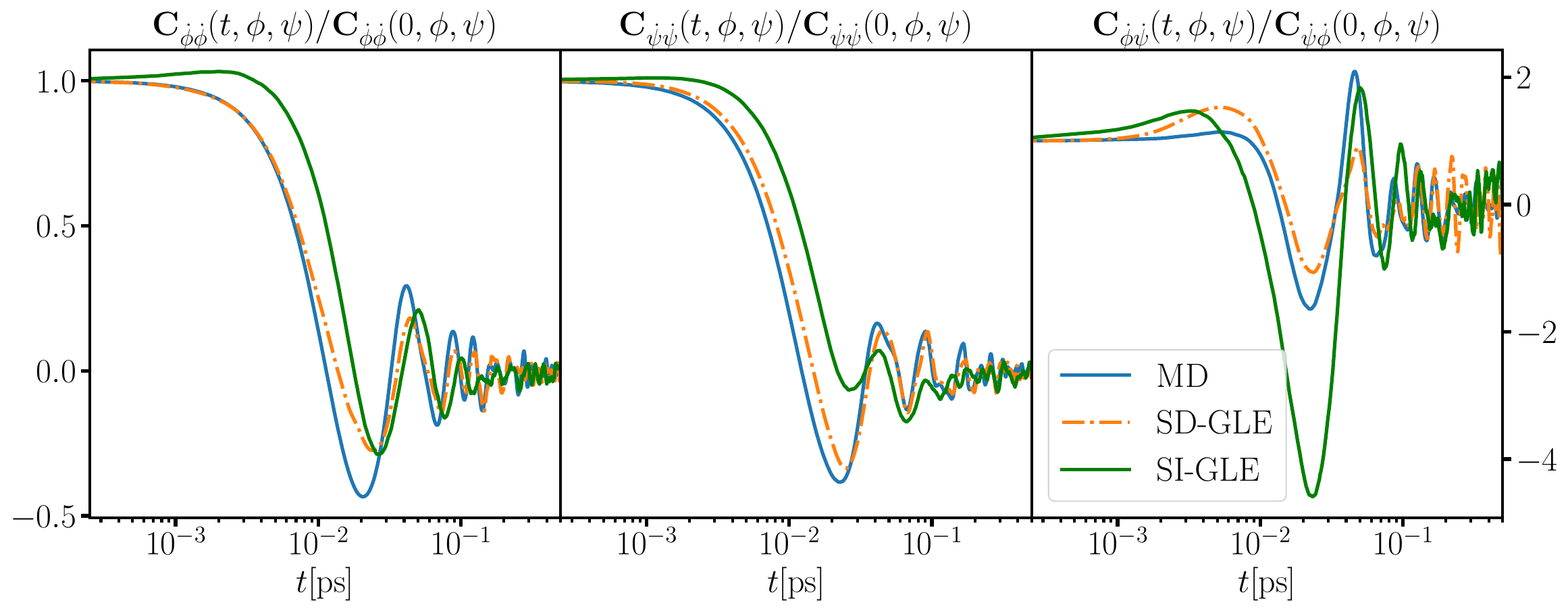}
    \caption{The conditional velocity correlation function $\mb C_{vv}(t, \phi, \psi)$ for both the diagonal components (left axis) and off-diagonal component (right axis) at a metastable point  $\phi=-1.594,\psi=2.785$.}
    \label{fig:2D_corr_123}
\end{figure}

\begin{figure}
    \centering
    \includegraphics[width=0.85\linewidth]{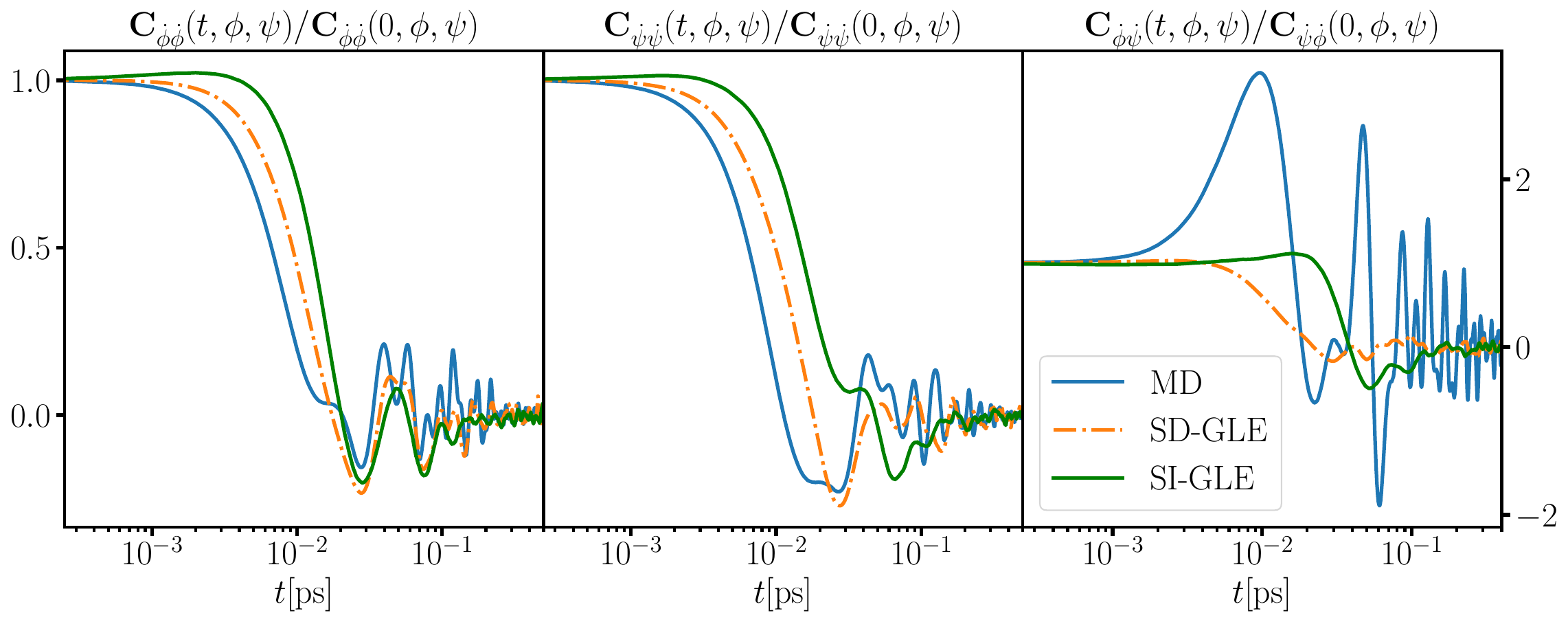}
    \caption{The conditional velocity correlation function $\mb C_{vv}(t, \phi, \psi)$ at a metastable point $\phi=-2.751,\psi=2.776$. The normalized off-diagonal component exhibits a large discrepancy due to a small magnitude of $C_{\dot{\phi}\dot{\psi}}(0, \phi, \psi) \sim \mathcal{O}(10^{-2})$.}
    \label{fig:2D_corr_772}
\end{figure}

\begin{figure}
    \centering
    \includegraphics[width=0.85\linewidth]{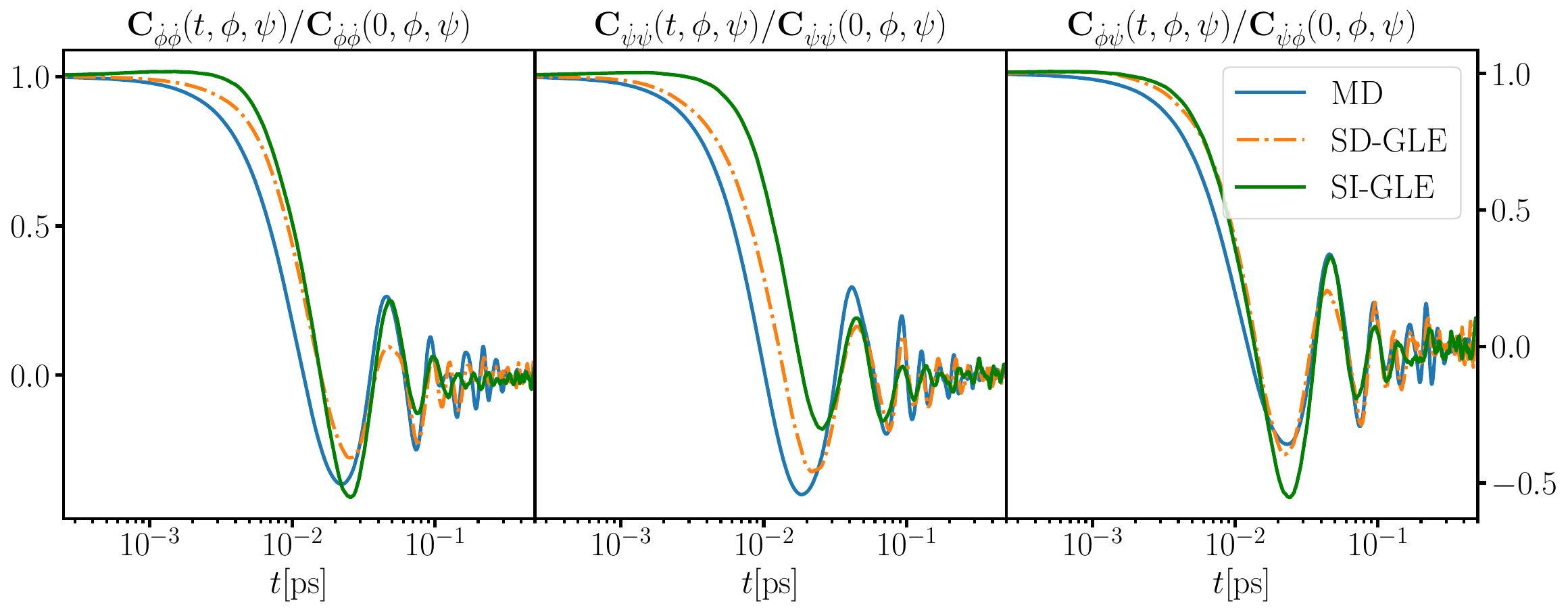}
    \caption{The conditional velocity correlation function $\mb C_{vv}(t, \phi, \psi)$ at a metastable point $\phi=-1.236,\psi=-0.578$.}
    \label{fig:2D_corr_999}
\end{figure}

\begin{figure}
    \centering
    \includegraphics[width=0.62\linewidth]{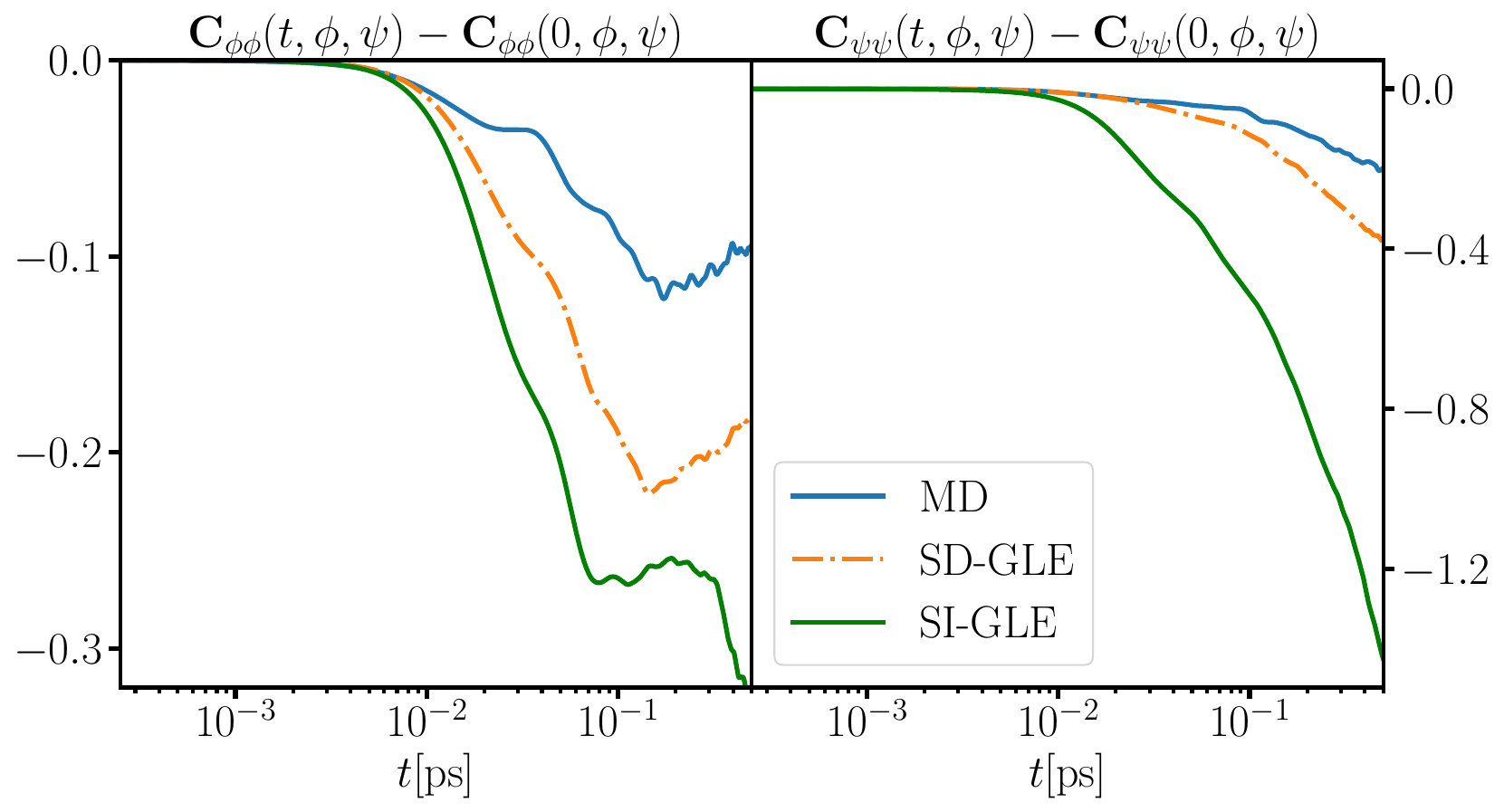}
    \caption{The conditional position correlation function $\mb C_{qq}(t, \phi, \psi)$ at a metastable point at $\phi=-1.594,\psi=2.785$.}
    \label{fig:2D_corr_qq_123}
\end{figure}

\begin{figure}
    \centering
    \includegraphics[width=0.62\linewidth]{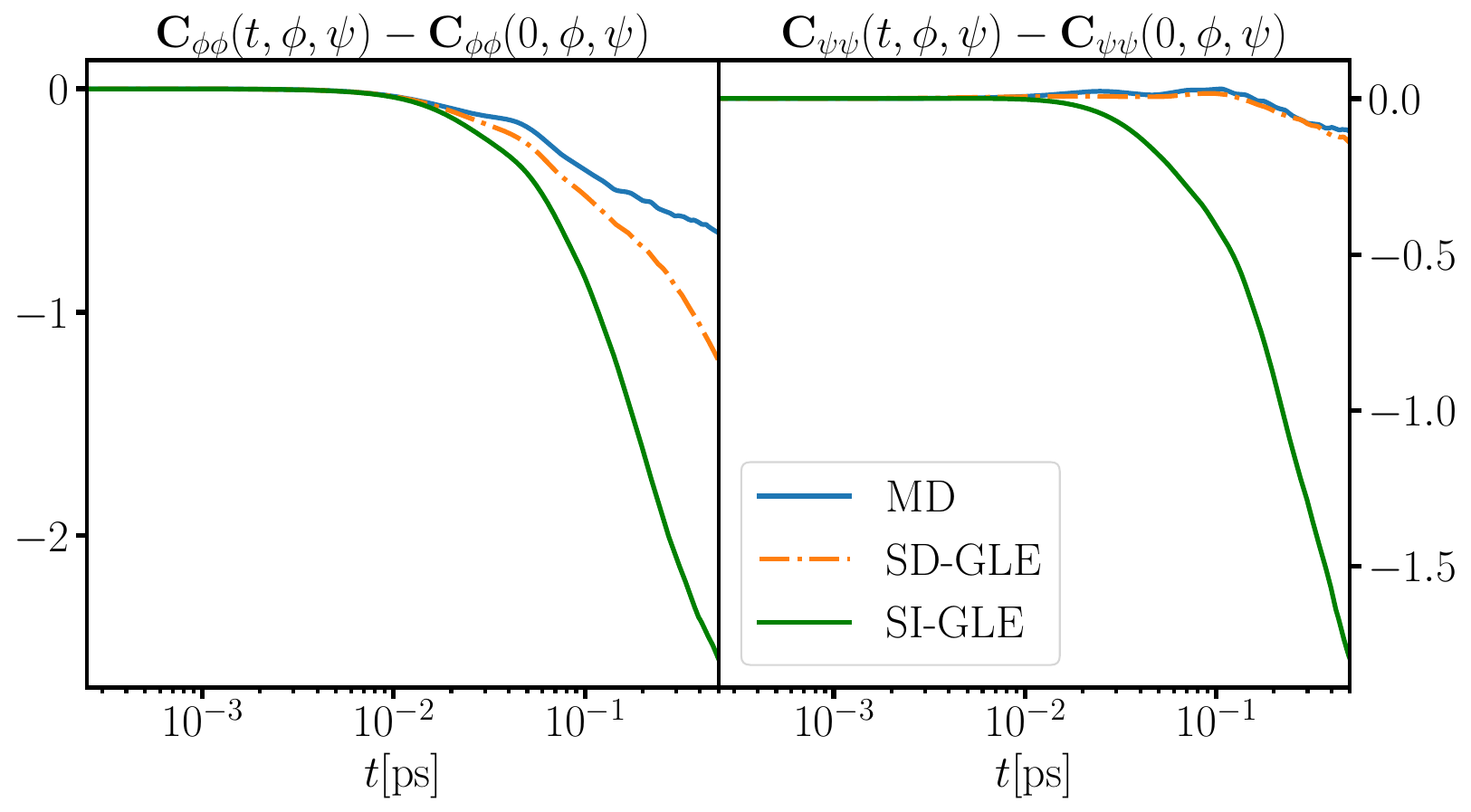}
    \caption{The conditional position correlation function $\mb C_{qq}(t, \phi, \psi)$ at a metastable point $\phi=-2.751,\psi=2.776$.}
    \label{fig:2D_corr_qq_772}
\end{figure}

\begin{figure}
    \centering
    \includegraphics[width=0.62\linewidth]{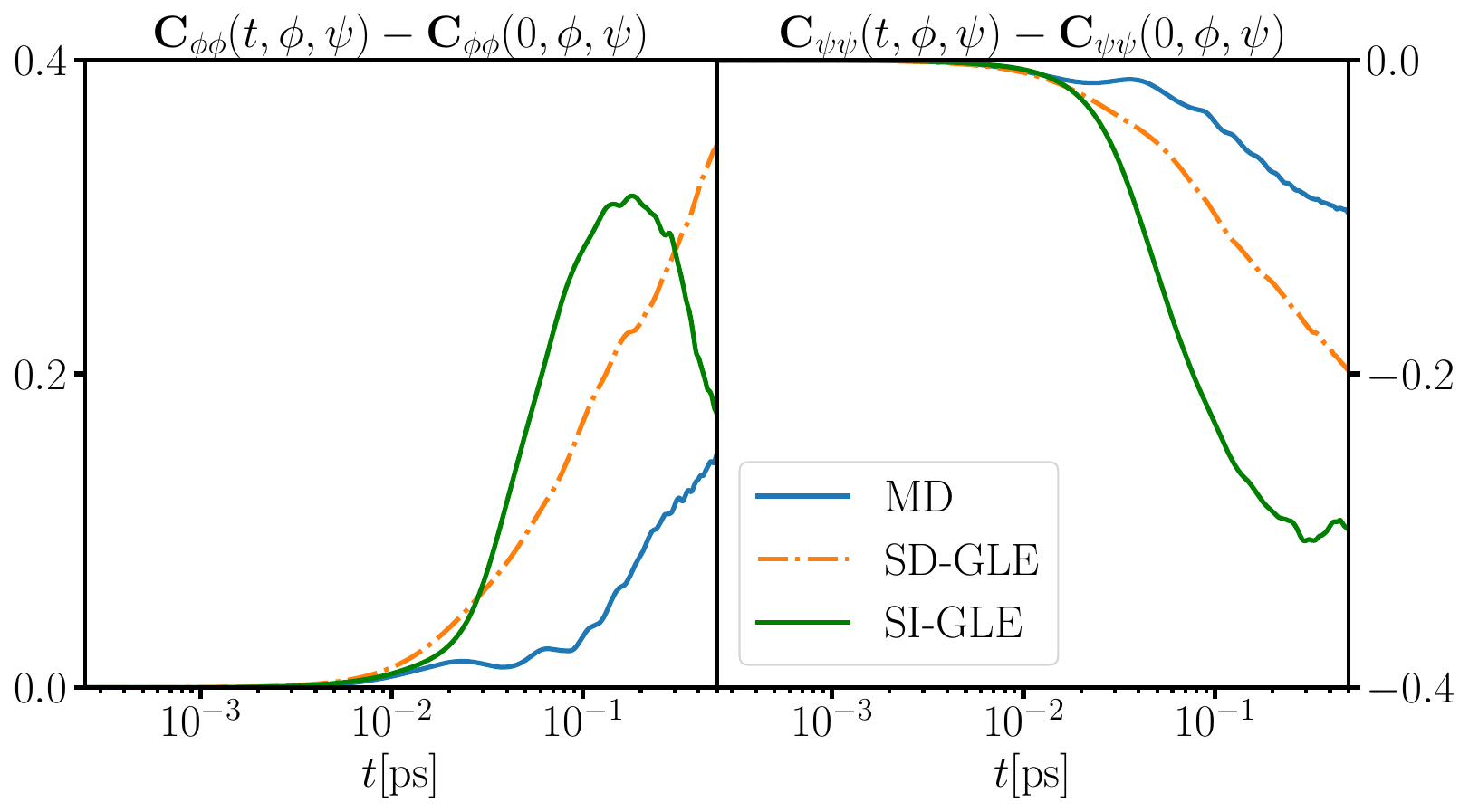}
    \caption{The conditional position correlation function $\mb C_{qq}(t, \phi, \psi)$ at a metastable point $\phi=-1.236,\psi=-0.578$.}
    \label{fig:2D_corr_qq_999}
\end{figure}

\begin{figure}
    \centering
    \includegraphics[width=0.88\linewidth]{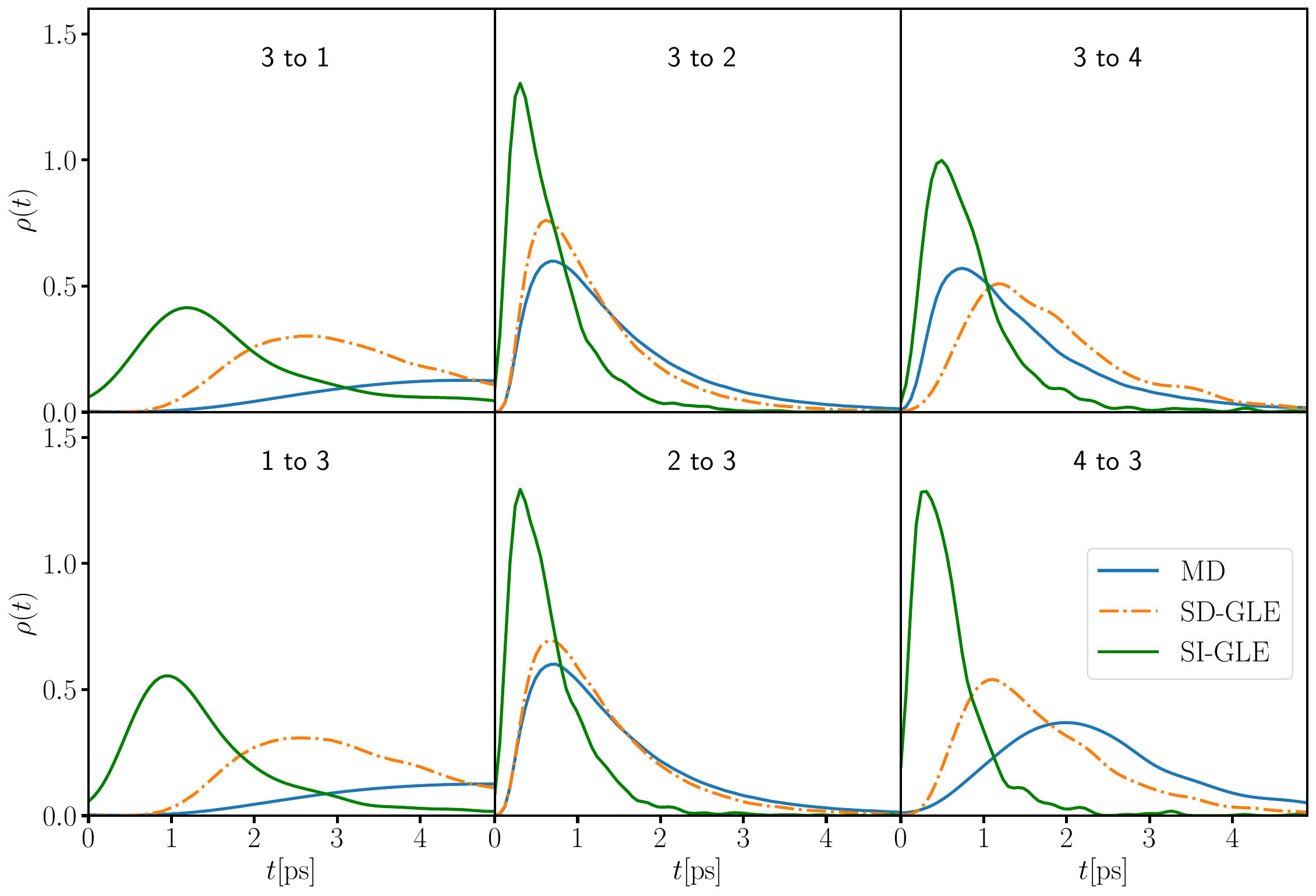}
    \caption{The distribution of the committor-based transition time (defined by Eq. \eqref{eq:transition_time}) that characterizes the 
    duration of the actual transition without re-crossing .}
    \label{fig:2D_transition}
\end{figure}

\appendix
\section{Simulation Setup}
\label{app:simulation}
\subsection{polymer molecule}
\label{app:polymer}
The polymer molecule is represented as a bead–spring chain consisting of four sub-units, each containing four atoms. The total potential energy is expressed as
\begin{equation}
V_{\text{mol}}(Q) = \sum_{i \neq j} V_p(Q_{ij}) 
+ \sum_{i=1}^{N_b} V_b(l_i) 
+ \sum_{i=1}^{N_a} V_a(\theta_i) 
+ \sum_{i=1}^{N_d} V_d(\phi_i),
\label{eq:Vtotal}
\end{equation}
where \(V_p, V_b, V_a,\) and \(V_d\) denote the pairwise, bond, angle, and dihedral interactions, respectively. The detailed functional forms and parameters are given below.

The pairwise interaction \(V_p\) is modeled by the truncated Lennard–Jones potential:
\begin{equation}
V_p(Q) = 
\begin{cases}
4\varepsilon \left[ \left(\dfrac{\sigma}{Q}\right)^{12} - \left(\dfrac{\sigma}{Q}\right)^{6} \right] 
- 4\varepsilon \left[ \left(\dfrac{\sigma}{Q_c}\right)^{12} - \left(\dfrac{\sigma}{Q_c}\right)^{6} \right], & Q < Q_c, \\
0, & Q \geq Q_c ,
\end{cases}
\end{equation}
with parameters \(\varepsilon = 0.005\), \(\sigma = 1.8\), and cutoff distance \(Q_c = 10.0\).

Bonded interactions are modeled by the finite extensible nonlinear elastic (FENE) potential:
\begin{equation}
V_b(l) = -\frac{k_s l_0^2}{2} 
\log \left( 1 - \frac{l^2}{l_0^2} \right),
\end{equation}
where three different bond types are defined. Within each sub-unit, the atoms 1-2, 3-4 are connected by type 1. The atoms 2-3 are connected by type-2 bond. Finally, the sub-unit groups are connected by type-3.
The parameter values are summarized in Table~\ref{tab:bonds}.
\begin{table}[h]
\centering
\caption{Parameters of the FENE bond interactions.}
\label{tab:bonds}
\begin{tabular}{c c c}
\hline
Type & $k_s$ & $l_0$ \\
\hline
1 & 0.40 & 1.8 \\
2 & 0.64 & 1.6 \\
3 & 0.32 & 1.8 \\
\hline
\end{tabular}
\end{table}

Angular interactions are described by a harmonic potential:
\begin{equation}
V_a(\theta) = \frac{k_a}{2} (\theta - \theta_0)^2 ,
\end{equation}
where two different types are used. Within each sub-unit group, the bond angles formed by 1-2-3 and 2-3-4 are imposed by type-1 potential. The bond angles formed by atoms of different sub-unit groups (e.g. 3-4-5,4-5-6) are imposed by type-2 potential. Parameters are listed in Table~\ref{tab:angles}.
\begin{table}[h]
\centering
\caption{Parameters of the harmonic angle interactions.}
\label{tab:angles}
\begin{tabular}{c c c}
\hline
Type & $k_a$ & $\theta_0$ (deg) \\
\hline
1 & 1.2 & 114.0 \\
2 & 1.5 & 119.7 \\
\hline
\end{tabular}
\end{table}

The dihedral potential is modeled by a multiharmonic expansion:
\begin{equation}
V_d(\phi) = \sum_{n=1}^{6} A_n \cos^{n-1}(\phi),
\end{equation}
where two different types are used. Type 1 dihedral potential is imposed to dihedral angles formed by 2–3–4–5, 4–5–6–7, $\cdots$, and Type 2 dihedral angle are imposed to dihedral angles formed by 3–4–5–6, 7–8–9–10, $\cdots$.  Parameters are given in Table~\ref{tab:dihedrals}.

\begin{table}[h]
\centering
\caption{Parameters of the multiharmonic dihedral interactions.}
\label{tab:dihedrals}
\begin{tabular}{c c c c c c c}
\hline
Type & $A_1$ & $A_2$ & $A_3$ & $A_4$ & $A_5$ & $A_6$ \\
\hline
1 & 0.0673 & 1.8479 & 0.0079 & -2.2410 & -0.0058 & 0.0051 \\
2 & 0.1602 & -3.9993 & 0.2483 & 6.2837 & 0.0165 & -0.0146 \\
\hline
\end{tabular}
\end{table}

\subsection{alanine dipeptide}
\label{app:ala_mol}
The MD simulations were performed with GROMACS 2019.2\cite{lindahl_2019_2636382}  in combination with the open-source PLUMED library~\cite{tribello2014plumed}. The system was modeled using the Amber99SB force field~\cite{hornak2006comparison}, with the alanine dipeptide solvated in an aqueous environment containing $383$ TIP3P water molecules. Periodic boundary conditions were applied in all three spatial directions. A cutoff of 0.9 nm was used for van der Waals interactions. Long-range electrostatics were treated with the smooth particle-mesh Ewald (PME) method, using a real-space cutoff of $0.9$ nm and a reciprocal-space grid spacing of $0.12$ nm. The equations of motion were integrated with the leap-frog scheme and a time step of $2$ fs. The system temperature was maintained at $300$ K using a velocity-rescale thermostat~\cite{bussi2007canonical} with a relaxation time of $0.2$ ps. Pressure was controlled at 1 bar with the Parrinello–Rahman barostat~\cite{parrinello1981polymorphic} employing a relaxation time of 1.5 ps and an isothermal compressibility of $4.5\times 10^{-5}$ bar $^{-1}$. Bond lengths involving hydrogen atoms were constrained with the LINCS algorithm~\cite{hess1997lincs}, while the H–O bond lengths and H–O–H angles in water molecules were constrained with the SETTLE algorithm~\cite{miyamoto1992settle}. All simulations were carried out on an Intel(R) Xeon(R) Platinum 8260 CPU.

\section{Training Details}
\label{app:training}
The NNs are trained by Adam optimizer \cite{Kingma_Ba_Adam_2015} for $100000$ steps with a learning rate that starts from $1\times10^{-3}$ and decays to $1\times10^{-5}$ in polynomial form. For each training step, $100$ sampling points are randomly selected from the data set. All the training processes are conducted using an Nvidia GPU V100 with 32GB memory.




\section*{Acknowledgments}
The work is supported in part by the National Science Foundation under Grant DMS-2110981, the
Strategic Partnership Grant at Michigan State University and the ACCESS program through allocation MTH210005. The authors also acknowledge the support from the Institute for Cyber-Enabled Research at Michigan State University.


\end{document}